\documentclass[12pt,journal,onecolumn, a4paper, draftclsnofoot]{IEEEtran}
\usepackage[utf8]{inputenc}
\usepackage{amsmath,amssymb,amsfonts,amsthm,mathrsfs,bm,bbm}
\usepackage{bm}
\usepackage{graphicx,overpic}
\usepackage[caption=false,font=footnotesize]{subfig}
\usepackage[shortlabels]{enumitem}
\usepackage{cite}
\usepackage{algorithm, algpseudocode}

\usepackage{hyperref}
\hypersetup{
  colorlinks   = true,    
  urlcolor     = blue,    
  linkcolor    = black,    
  citecolor    = red      
}

\newtheorem{lemma}{Lemma}
\newtheorem{definition}{Definition}

\newtheorem{theorem}{Theorem}

\newtheorem{assumption}{Assumption}
\newtheorem{remark}{Remark}
\newtheorem{corollary}{Corollary}

\newcommand{\rvec}[1]{\mathbbm{#1}}  	
\newcommand{\E}{\mathsf{E}}				
\newcommand{\V}{\mathsf{Var}}			
\newcommand{\stdset}[1]{\mathbbmss{#1}}	
\newcommand{\set}[1]{\mathcal{#1}}		
\renewcommand{\vec}[1]{\mathbf{#1}}		
\newcommand{\CN}{\mathcal{CN}}			
\newcommand{\herm}{\mathsf{H}}			
\newcommand{\T}{\mathsf{T}}				

\author{Lorenzo Miretti, Emil Bj\"ornson, and David Gesbert\thanks{A preliminary version of this work has been presented at IEEE SPAWC 2021  \cite{miretti2021precoding}.}\thanks{This work received partial support from the Huawei funded Chair on Future Wireless Networks at EURECOM.}\thanks{E.~Bj\"ornson was supported by the Grant 2019-05068 from the Swedish Research Council.}\thanks{This work has been submitted to the IEEE for possible publication.  Copyright may be transferred without notice, after which this version may no longer be accessible.}}
\title{Team MMSE Precoding with Applications to Cell-free Massive MIMO}

\begin{document}
\maketitle
\begin{abstract}
This article studies a novel distributed precoding design, coined \textit{team minimum mean-square error} (TMMSE) precoding, which rigorously generalizes classical centralized MMSE precoding to distributed operations based on transmitter-specific channel state information (CSIT). Building on the so-called \textit{theory of teams}, we derive a set of necessary and sufficient conditions for optimal TMMSE precoding, in the form of an infinite dimensional linear system of equations. These optimality conditions are further specialized to cell-free massive MIMO networks, and explicitly solved for two important examples, i.e., the classical case of local CSIT and the case of  unidirectional CSIT sharing along a serial fronthaul. The latter case is relevant, e.g., for the recently proposed \textit{radio stripe} concept and the related advances on sequential processing exploiting serial connections. In both cases, our optimal design outperforms the heuristic methods that are known from the previous literature. Duality arguments and numerical simulations validate the effectiveness of the proposed team theoretical approach in terms of ergodic achievable rates under a sum-power constraint. 
\end{abstract}
 
\section{Introduction}
Inter-cell interference is a major limiting factor of wireless communication systems capitalizing on aggressive spectrum reuse and network densification to increase capacity. To mitigate this effect, future generation systems are expected to implement advanced cooperative communication techniques, in particular by letting geographically distributed base stations jointly serve their users. However, the practical deployment of cooperative wireless networks is currently prevented by the severe scalability issue arising from network-wide processing~\cite{gesbert2010multicell}. Specifically, the excessive amount of data and channel state information (CSI) to be timely shared for implementing fully cooperative regimes such as in the original \textit{network MIMO} or \textit{cloud radio access network} (C-RAN) concepts \cite{shamai2001enhancing,venkatesan2007network,checko2015cran} often becomes the main bottleneck when practical fronthaul capacity constraints are introduced. Studying more realistic cooperation regimes entailing limited data and CSI sharing is hence of fundamental importance for making network cooperation an attractive technology for next generation systems \cite{gesbert2010multicell,zhang2020prospective}. 

\subsection{Cooperative transmission with distributed CSIT}
In this work we explore a downlink (DL) cooperation regime with full data sharing and general \textit{distributed} CSI at the TXs (CSIT) \cite{dekerret2012degrees,gesbert2018team,bazco2020degrees}, that is, we let each TX operate on the basis of possibly different estimates of the \textit{global} channel state obtained through some arbitrary CSIT acquisition and sharing mechanism. This assumption is relevant, e.g.,  for all service situations where, compared to data sharing, CSIT sharing needs to be performed within much tighter time constraints, and hence may dominate the fronthaul overhead. For instance, it is suitable in case of rapidly varying channels due to user mobility, where full CSIT sharing may result in outdated information or occupy an excessive portion of the coherence time, or when delay-tolerant data is proactively made available at the TXs using caching techniques (see \cite{bazco2020degrees} and reference therein for a detailed discussion). As an extreme example, a cooperation regime with full data sharing and 
no CSIT sharing (a configuration here referred to as \textit{local} CSIT) is perhaps the leading motivation behind the early development of the now popular \textit{cell-free massive MIMO} paradigm \cite{ngo2017cellfree}. This paradigm combines the benefits of ultra-dense networks with simple yet effective TX cooperation schemes, and emerged as a promising evolution of the network MIMO and C-RAN concepts. The distributed CSIT assumption also covers extensions of \cite{ngo2017cellfree} to more complex setups ranging from partial to full CSIT sharing (see, e.g., \cite{nayebi2017precoding} \cite{du2021cellfree}). Clearly, these cooperation regimes are still far from being scalable, since they all assume network-wide data sharing. Splitting the network into clusters of cooperating TXs \cite{huang2009increasing,zhang2009networked,du2021cellfree}, possibly dynamically and with a user-centric approach \cite{emil2013optimal,buzzi2020usercentric,interdonato2019ubiquitous,emil2020scalable}, and applying similar transmission techniques assuming full data sharing  within each cluster, emerged as a viable solution for implementing scalable cooperation regimes in practical systems. However, due to space limitations and to better focus on limited CSIT sharing, in this work, we do not consider network clusterization. In particular, we do not cover complementary service situations where CSIT can be more easily shared than data, hence entering the realm of interference \textit{coordination} or \textit{alignment} \cite{gesbert2010multicell,cadambe2008interference}.
Nevertheless, if combined with network clustering techniques, the results presented in this study can be seen as a first step towards a more general theory jointly covering limited data and CSIT sharing.


\subsection{Summary of contributions and related works}
Although the importance of cooperative transmission schemes based on limited CSIT sharing has been acknowledged in the literature, a satisfactory understanding of systems with distributed CSIT is still missing. Most of the available information theoretical results rely on asymptotic signal-to-noise ratio (SNR) tools\cite{dekerret2012degrees,bazco2020degrees}, or focus on simple settings with a single receiver (RX) \cite{miretti2019cooperative}. However, \cite{dekerret2012degrees,bazco2020degrees,miretti2019cooperative}  do not lead to practical schemes for complex settings such as cell-free massive MIMO networks. On the other hand, the available practical schemes are essentially based on heuristic adaptations
of known centralized precoding designs such as maximum-ratio
transmission (MRT), zero-forcing (ZF), or minimum mean-square
error (MMSE) precoding \cite{ngo2017cellfree,nayebi2017precoding,emil2020scalable}.
Hence, there is a need to develop a mathematical framework that allows the sound derivation of schemes that cater for distributed CSIT setups.

This work provides considerable progress in this direction. By using simplified yet practical point-to-point information theoretical tools, namely by using standard linearly precoded Gaussian codes, treating interference as noise, and the non-coherent ergodic rate bounds popularized by the massive MIMO literature \cite{marzetta2016fundamentals,massivemimobook}, we propose a novel distributed precoding design, coined \textit{team} MMSE (TMMSE) precoding, generalizing classical centralized MMSE precoding \cite{massivemimobook} to systems with distributed CSIT. Its optimality in terms of achievable ergodic rates under a sum-power constraint is formally established by revisiting the uplink-downlink (UL-DL) duality principle \cite{massivemimobook} in light of the distributed CSIT assumption. Our first main result\footnote{The preliminary version of this work \cite{miretti2021precoding} focuses on a simplified cell-free massive MIMO setup. This work extends \cite{miretti2021precoding} to more general networks, Gaussian fading, and channel estimation errors; provides complete theoretical derivations; improves the comparison with the previous literature.} is showing that the problem of optimal TMMSE precoding design can be solved by means of a useful set of optimality conditions in the form of an infinite dimensional linear system of equations, for which many standard solution tools exist. The key novelty lies in the introduction of previously unexplored elements from the \textit{theory of teams}, a mathematical framework for multi-agent coordinated decision making in presence of asymmetry of information. This framework was pioneered in theoretical economics by Marschak and Radner \cite{marschak1972economic,radner1962team}, and then further developed in the control theoretical literature (see the excellent survey in \cite{yukselbook}). Early applications of team theory to wireless communication, including the problem of distributed precoding design, are reported in \cite{gesbert2010multicell,gesbert2018team}. However, compared to previous attempts for distributed precoding design, this work is the first exploiting (and partially extending) known results for the class of \textit{quadratic teams} \cite{radner1962team,yukselbook}, which is one of the few cases where solid globally optimal solution approaches are available. 

In the second part of this work, the aforementioned optimality conditions are specialized to cell-free massive MIMO networks. To the best of the authors' knowledge, this is the first work connecting cell-free massive MIMO to the theory of teams. The first non-trivial application is the derivation of the optimal TMMSE precoders based on local CSIT only, improving upon previous local precoding strategies studied, e.g., in \cite{ngo2017cellfree,neumann2018bilinear,emil2020scalable}. We then consider a cell-free massive MIMO network with serial fronthaul, an efficient architecture also known as a \textit{radio stripe} \cite{interdonato2019ubiquitous,shaik2020mmse}. We derive optimal TMMSE precoders by assuming that CSIT is shared unidirectionally along the stripe. The proposed scheme can be efficiently implemented in a sequential fashion, an idea that has been explored already in \cite{interdonato2019ubiquitous,emil2020cellfree,shaik2020mmse} for UL processing, and in \cite{rodriguez2020decentralized} under a different cellular context. As a byproduct, we also obtain a novel distributed implementation of classical centralized MMSE precoding tailored to radio stripes. Finally, we present extensive numerical results comparing the effects of different CSIT sharing patterns in a radio stripe system and evaluating the suboptimality of the competing schemes. Interestingly, our numerical results suggest that unidirectional information sharing is a promising candidate for enlarging the domain of applications of radio stripes beyond the regimes supported by centralized or local precoding; for instance, it may allow effective interference management for a wider range of mobility patterns. Moreover, we show that the known local MMSE precoding scheme studied, e.g., in \cite{emil2020scalable,emil2020cellfree}, is optimal in a non line-of-sight (NLoS) scenario, while it may be significantly outperformed by the TMMSE solution for local CSIT in the presence of line-of-sight (LoS) components. 

\subsection{Outline and notation}
The present study is structured as follows: Section~\ref{sec:model} presents the system model and other necessary preliminaries. The main results on team MMSE precoding and their specialization to cell-free massive MIMO networks are given respectively in Section~\ref{sec:TMMSE} and Section~\ref{sec:applications}. The numerical results are given in Section~\ref{sec:simulations}. 
 
Hereafter, we use the following notation. We reserve italic letters (e.g., $a$) for scalars and functions, boldface letters (e.g., $\vec{a}$, $\vec{A}$) for vectors and matrices, and calligraphic letters (e.g., $\mathcal{A}$) for sets. Random quantities are distinguished from their realizations as follows: $\rvec{a}$, $\rvec{A}$ denote random vectors and matrices; $A$ denotes a random scalar, or a generic random variable taking values in some unspecified set $\set{A}$. We use $:=$ for definitions, and $\preceq$, $\succeq$ ($\prec$, $\succ$) for (strict) generalized inequalities w.r.t. the cone of nonnegative Hermitian matrices. The operators $(\cdot)^\T$, $(\cdot)^\herm$ denote respectively the transpose and Hermitian transpose of matrices and vectors, and $\Re(\cdot)$ is the real part. We denote the Euclidean norm by $\|\cdot\|$, the Frobenius norm by $\|\cdot\|_\mathrm{F}$, and the trace operator by $\mathrm{tr}(\cdot)$. The (conditional) expectation of $\rvec{A}$ (given $\rvec{B}$) is denoted by $\E[\rvec{A}]$ ($\E[\rvec{A}|\rvec{B}]$), and $\V[A]$ is the variance of~$A$. Given $n>2$ random matrices $\rvec{A}_1,\ldots,\rvec{A}_n$ with joint distribution $p(\vec{A}_1,\ldots,\vec{A}_n)$, we say that $\rvec{A}_1\to \rvec{A}_2 \to \ldots \to \rvec{A}_n$ forms a Markov chain if $p(\vec{A}_i|\vec{A}_{i-1},\ldots,\vec{A}_1) = p(\vec{A}_i|\vec{A}_{i-1})$ $\forall i \geq 2$. We use $\mathrm{diag}(\vec{A}_1,\ldots,\vec{A}_n)$ to denote a block-diagonal matrix with the matrices $\vec{A}_1,\ldots,\vec{A}_n$ on its diagonal, and $\mathrm{vec}(\vec{A})$ to denote a vector obtained by stacking column-wise the elements of $\vec{A}$. We denote by $\vec{e}_n$ the $n$-th column of the identity matrix $\vec{I}$, with dimension extrapolated from the context. We use $\prod_{i=l'}^{l}\vec{A}_i := \vec{A}_{l}\vec{A}_{l-1}\ldots\vec{A}_{l'}$ for integers $l \geq l' \geq 1$ to denote the \textit{left} product chain of $l-l'+1$ ordered matrices of compatible dimension, and we adopt the convention $\prod_{i=l'}^{l} \vec{A}_i = \vec{I}$ for $l<l'$. Finally, $h(A)$ ($h(A|B)$) denotes the (conditional) entropy, $I(A;B)$ is the mutual information, and all logarithms are expressed in base $2$ unless differently specified.

\section{System model and preliminaries}
\label{sec:model}
\subsection{Channel model}
Consider a network of $L$ TXs indexed by $\set{L}:=\{1,\ldots,L\}$, each of them equipped with $N$ antennas, and $K$ single-antenna RXs indexed by $\set{K}:=\{1,\ldots,K\}$. Let an arbitrary channel use be governed by the MIMO channel law
\begin{equation}
\rvec{y} = \sum_{l=1}^L \rvec{H}_l\rvec{x}_l + \rvec{n} 
\end{equation}
where the $k$-th element $Y_k$ of $\rvec{y} \in \stdset{C}^{K}$ is the received signal at RX $k$, $\rvec{H}_l \in \stdset{C}^{K\times N}$ is a sample of a stationary ergodic random process modelling the fading between TX $l$ and all RXs, $\rvec{x}_l \in \stdset{C}^N$ is the transmitted signal at TX~$l$, and $\rvec{n}\sim \CN(\vec{0},\vec{I})$ is a sample of a white noise process. This channel model is relevant, e.g, for narrowband or wideband OFDM systems \cite{tse2005fundamentals} where transmission spans several realizations of the fading process. For most parts of this work, we do not specify the distribution of $\rvec{H}:=\begin{bmatrix}
\rvec{H}_1,\ldots,\rvec{H}_L
\end{bmatrix}$. However, we reasonably assume the channel submatrices corresponding to different TX-RX pairs to be mutually independent, and finite fading power $\E[\|\rvec{H}\|_{\mathrm{F}}^2]<\infty$. Furthermore, we focus on $N<K$, that is, on the regime where cooperation is crucial for interference management \cite{gesbert2010multicell}.

\subsection{Distributed linear precoding}
Consider a \textit{distributed} CSIT configuration \cite{gesbert2018team}, i.e., where each TX has some potentially different side information $S_l$ about the \textit{global} channel matrix $\rvec{H}$. For instance, this could model frequency division duplex (FDD) systems where each $S_l$ is composed by different feedback signals from the RXs, or time-division duplex (TDD) systems where over-the-uplink \textit{local} estimates $\hat{\rvec{H}}_l$ of the \textit{local} channel $\rvec{H}_l$ are not perfectly shared across the network. Importantly, this is more general than the typical assumption in the cell-free massive MIMO literature covering distributed operations, which limits $S_l$ to $\hat{\rvec{H}}_l$ only \cite{ngo2017cellfree,emil2020scalable}. We assume $(\rvec{H},S_1,\ldots,S_L)$ to be a sample of an ergodic stationary process with first order joint distribution fixed by nature/design, and known by all TXs. 

We then let each TX $l$ form its transmit signal according to the following distributed linear precoding scheme:
\begin{equation}\label{eq:distributed_precoding}
\rvec{x}_l = \sum_{k=1}^K\rvec{t}_{l,k}U_k,\quad \rvec{t}_{l,k} = \rvec{t}_{l,k}(S_l),
\end{equation}
where $U_k \sim \CN(0,p_k)$ is the independently encoded message for RX $k$, shared by all TXs, and where $\rvec{t}_{l,k}\in \stdset{C}^N$ is a linear precoder applied at TX $l$ to message $U_k$ based only on the side information $S_l$. More formally, by letting $(\Omega,\Sigma,\mathsf{P})$ be the underlying probability space over which all random quantities are defined, we constrain $\rvec{t}_{l,k}$ within the vector space $\set{T}_l$ of square-integrable $\Sigma_l$-measurable functions $\Omega \to \stdset{C}^N$, where $\Sigma_l\subseteq \Sigma$ denotes the sub-$\sigma$-algebra generated by $S_l$ on $\Omega$, called the \textit{information subfield} of TX $l$ \cite{radner1962team,yukselbook}. This assumption\footnote{The measure theoretical formulation presented above is necessary for establishing Theorem \ref{th:quadratic_teams_new}. However, the rest of this study does not require any particular measure theoretical background.} rigorously describes the functional dependency of $\rvec{t}_{l,k}$ on the portion $S_l$ of the overall system randomness, and includes a reasonable finiteness constraint $\E[\|\rvec{t}_{l,k}\|^2]<\infty$ on precoders power. We finally denote the full precoding vector for message $U_k$ by $\rvec{t}_k^\T := \begin{bmatrix}
\rvec{t}_{1,k}^\T & \ldots & \rvec{t}_{L,k}^\T
\end{bmatrix}^\T$, and let $\rvec{t}_k \in \set{T}:=\prod_{l=1}^L\set{T}_l$.

\subsection{Performance metric}
We measure the network performance under the specified transmission scheme by using Shannon (ergodic) rates $R_k^{\mathrm{DL}} := I(U_k;Y_k)$, which are achievable without channel state information at the RX (CSIR) by treating interference as noise (TIN) and by neglecting  any memory across the realizations $(\rvec{H},S_1,\ldots,S_L)$ of the state and CSIT process \cite{biglieri1998fading,caire2018ergodic}. Because of the difficulties in evaluating the mutual information, we consider the following lower bound, known as the \textit{hardening} bound \cite{marzetta2016fundamentals,massivemimobook},
\begin{equation}\label{eq:hardening_bound}
I(U_k;Y_k) \geq  \log\left(1+ \dfrac{p_k|\E[\rvec{g}_k^\herm\rvec{t}_k]|^2}{p_k\V[\rvec{g}_k^\herm\rvec{t}_k] + \sum_{j\neq k}p_j\E[|\rvec{g}_k^\herm\rvec{t}_j|^2]+1}\right)=: R_k^{\mathrm{hard}},
\end{equation}
where $\begin{bmatrix}
\rvec{g}_1 & \ldots & \rvec{g}_K \end{bmatrix} :=\rvec{H}^\herm
$. An alternative classical performance metric would be given by the following upper bound
\begin{equation}\label{eq:R_up}
I(U_k;Y_k) \leq  \E\left[\log\left( 1+\dfrac{p_k|\rvec{g}_k^\herm\rvec{t}_k|^2}{\sum_{j\neq k}p_j|\rvec{g}_k^\herm\rvec{t}_j|^2+1}\right)\right]=:R_k^{\mathrm{ub}},
\end{equation}
which is in fact achievable with perfect local CSIR $\{\rvec{g}_k^\herm\rvec{t}_j\}_{j=1}^K$ by TIN and by taking into account channel memory \cite{biglieri1998fading,caire2018ergodic}. We consider $R_k^{\mathrm{hard}}$ instead of $R_k^{\mathrm{ub}}$ because of the less stringent CSIR requirements and, perhaps most importantly, for treatability reasons. Due to its name and historical use, it is sometimes believed that the ability of $R_k^{\mathrm{hard}}$ to produce good approximations of $R_k^{\mathrm{DL}}$ relies on the \textit{channel hardening} effect arising in massive MIMO systems~\cite{massivemimobook}. Although this is correct for some precoding design such as MRT \cite{caire2018ergodic}, we remark that \eqref{eq:hardening_bound} may perform well also in absence of channel hardening. For instance, $R_k^{\mathrm{hard}}$ and $R_k^{\mathrm{up}}$ coincide under a ZF scheme with perfect CSIT putting $\rvec{g}^\herm_k\rvec{t}_j = 0$ for $j\neq k$ and $\rvec{g}_k^\herm\rvec{t}_k = 1$, for any long-term power allocation policy $\{p_k\}_{k=1}^K$ and feasible antenna regime. 

We then let $\set{R}^{\mathrm{hard}}$ be the union of all rate tuples $(R_1,\ldots,R_K)\in \stdset{R}_+^K$ such that $R_k \leq R_k^{\mathrm{hard}}$ $\forall k \in \set{K}$ for some set of distributed precoders $\{\rvec{t}_k\}_{k=1}^K$ and power allocation policy $\{p_k\}_{k=1}^K$ satisfying $\sum_{l=1}^L\E[\|\rvec{x}_l\|^2]\leq P_{\mathrm{sum}}<\infty$. The set $\set{R}^{\mathrm{hard}}$ is an inner bound for the capacity region of the considered network with distributed CSIT and subject to a long-term sum power constraint $P_{\mathrm{sum}}$. Due to its importance in system design and resource allocation, we consider the notion of (weak) Pareto optimality on $\set{R}^{\mathrm{hard}}$ and we mostly focus on the (weak) Pareto boundary of $\set{R}^{\mathrm{hard}}$, denoted  by $\partial \set{R}^{\mathrm{hard}}$ \cite{emil2013optimal}. Note that, by the Cauchy–Schwarz inequality and the mild assumptions given in the previous sections, we have $|\mathbb{E}[\rvec{g}_k^\herm\rvec{t}_k]|^2< \infty$ $\forall \rvec{t}_k \in \set{T}$, $\forall k \in \set{K}$, hence $\partial \set{R}^{\mathrm{hard}}$ is finite.

The long-term sum power constraint is chosen because it allows for strong analytical results and simplifies system design. This constraint may be directly relevant for systems such as the radio stripes, treated in Section \ref{ssec:unidirectional}, where all the TXs share the same power supply \cite{interdonato2019ubiquitous}. However, note that many simple heuristic methods (such as power scaling factors) can be applied to adapt systems designed under a long-term sum power constraint to more restrictive cases such as per-TX power constraints. Further analyses on different power constraints are left for future work.

\section{Team MMSE precoding}\label{sec:TMMSE}
In this work, we study the following novel \textit{team} MMSE precoding design criterion: given a vector of nonnegative weights $\vec{w}:=[w_1,\ldots,w_K]^\T$ belonging to the simplex $ \set{W}:=\{\vec{w} \in \stdset{R}_+^K \; | \; \sum_{k=1}^Kw_k = K\}$,  we consider the functional optimization problem
\begin{equation}\label{eq:team_prob}
\underset{\rvec{t}_k \in \set{T}}{\text{minimize}} \; \mathrm{MSE}_k(\rvec{t}_k) := \E\left[ \left\|\vec{W}^{\frac{1}{2}}\rvec{H}\rvec{t}_k - \vec{e}_k\right\|^2 +\dfrac{\|\rvec{t}_k\|^2}{P} \right],
\end{equation}
where $\vec{W}:=\mathrm{diag}(w_1,\ldots,w_K)$,  $\vec{e}_k$ is the $k$-th column of $\vec{I}_K$, and $P:=P_{\mathrm{sum}}/K$. A solution to the above problem can be recognized as a distributed version of the classical centralized MMSE precoding design \cite{massivemimobook}. For $P\to \infty$, it can be interpreted as the `closest' distributed approximation of the ZF solution. By means of team theoretical arguments \cite{radner1962team,yukselbook}, this section provides rigorous yet practical guidelines for optimally solving Problem~\eqref{eq:team_prob}. Before providing the main results of this section, we also revisit the effectiveness of the MSE criterion in terms of network performance, which is well-known for centralized precoding.

\begin{remark}
Hereafter, with the exception of Section \ref{ssec:rates}, we consider w.l.o.g. $\vec{W}=\vec{I}$. The general case will readily follow by replacing $\rvec{H}_l$ with $\vec{W}^{\frac{1}{2}}\rvec{H}_l$ everywhere. 
\end{remark}

\subsection{Achievable rates via uplink-downlink duality}	
\label{ssec:rates}
This section discusses the formal connection between the objective of Problem~\eqref{eq:team_prob} and $\set{R}^{\mathrm{hard}}$ by revisiting UL-DL duality under a general distributed CSIT assumption. 
\begin{theorem}\label{th:duality_WMSE} Consider an arbitrary set of distributed precoders $\{\rvec{t}_k\}_{k=1}^K$ and weights $\vec{w}\in \set{W}$. Then, any rate tuple $(R_1,\ldots,R_K) \in \stdset{R}^K$ such that
\begin{equation}\label{eq:WMSE_lower_bound}
R_k \leq \log(\mathrm{MSE}_k(\rvec{t}_k))^{-1}  
\end{equation}
belongs to $\set{R}^{\mathrm{hard}}$. The power allocation policy $\{p_k\}_{k=1}^K$ achieving the above inner bound is given in Appendix~ \ref{proof:duality}, and satisfies $\sum_{l}\E[\|\rvec{x}_l\|^2]= P_{\mathrm{sum}}$. Furthermore, if $\rvec{t}_k$ solves Problem~\eqref{eq:team_prob} $\forall k\in\set{K}$, then $(R_1,\ldots,R_K)$ with $R_k =\log(\mathrm{MSE}_k(\rvec{t}_k))^{-1}$ is Pareto optimal, and every rate tuple in $\partial \set{R}^{\mathrm{hard}}$ is obtained for some $\vec{w}\in \set{W}$.
\end{theorem}
\begin{proof}
The proof is based on connecting Problem \eqref{eq:team_prob} to the problem of ergodic rate maximization in a dual UL channel, where $\vec{w}$ is an UL power allocation vector, $\rvec{t}_k$ is a distributed UL combiner, and where achievable rates are measured by using the so-called \textit{use-and-then-forget} (UatF) bound \cite[Theorem~4.4]{massivemimobook}. The details are given in Appendix~\ref{proof:duality}. 
\end{proof}
Theorem \ref{th:duality_WMSE} states that the Pareto boundary of $\set{R}^{\mathrm{hard}}$ can be parametrized by $K-1$ nonnegative real parameters, i.e., by the weights $\vec{w}\in\set{W}$. A similar parametrization was already known for deterministic channels (see, e.g., \cite{emil2013optimal}), or, equivalently, for fading channels with perfect CSIT and CSIR. This work extends the aforementioned results to imperfect and possibly distributed CSIT, and no CSIR. In theory, $\vec{w}$ should be selected according to some network utility (e.g., the sum-rate or the max-min rate). In practice, $\vec{w}$ is often fixed heuristically (e.g., from the real UL powers), while the network utility is optimized a posteriori by varying the DL power allocation policy $\{p_k\}_{k=1}^K$. 

From a precoding design point of view, Theorem \ref{th:duality_WMSE} generalizes the duality-based argument behind classical MMSE precoding given by \cite{massivemimobook}. While \cite{massivemimobook} motivates the MMSE solution as the optimal combiner maximizing a dual UL ergodic rate bound based on coherent decoding, the proof of Theorem~\ref{th:duality_WMSE} directly relates the MSE criterion to the more conservative UatF bound. This last point is particularly relevant under distributed CSIT, where an optimal solution to the coherent ergodic rate maximization problem is not known in general. We recall that, in turn, \cite{massivemimobook} generalizes classical duality-based arguments for deterministic channels to fading channels. As a concluding remark, we stress that the inner bound \eqref{eq:WMSE_lower_bound} should not be confused with the well-known inner bound based on the notion of (weighted) MSE on the DL channel \cite{christensen2008weighted}, where the precoders for \textit{all} messages contribute to each rate bound.  

\subsection{Quadratic teams for distributed precoding design}
\label{ssec:quadratic_teams} 
Problem~\eqref{eq:team_prob} belongs to the known family of \textit{team decision} problems \cite{radner1962team,yukselbook}, which are generally difficult to solve for general information constraints $\rvec{t}_k \in \set{T}$. However, by rewriting the objective as $\mathrm{MSE}_k(\rvec{t}_k)=\E[c_k(\rvec{H},\rvec{t}_{1,k},\ldots,\rvec{t}_{L,k})]$, 
\begin{equation}\label{eq:quadratic_teams}
c_k(\vec{H},\vec{t}_{1,k},\ldots,\vec{t}_{L,k}):=\vec{t}_k^\herm \vec{Q}\vec{t}_k -2\Re\left(\vec{g}_k^\herm\vec{t}_k\right) +1,
\end{equation}
where $\rvec{Q}:= \rvec{H}^\herm\rvec{H} + \frac{1}{P}\vec{I}$, $\rvec{g}_k = \rvec{H}^\herm\vec{e}_k$, and by noticing that $\rvec{Q}\succ \vec{0}$ a.s., we recognize that Problem~\eqref{eq:team_prob} belongs to the class of \textit{quadratic teams} as defined in \cite[Sect.~4]{radner1962team}. This class exhibits strong structural properties, in particular related to the following solution concept:
\begin{definition}[Stationary solution \cite{yukselbook}]
A solution $\rvec{t}_k^\star \in \set{T}$ is a stationary solution for Problem~\eqref{eq:team_prob} if $\mathrm{MSE}_k(\rvec{t}_k^\star)<\infty$ and if the following set of equalities hold 
\begin{equation}\label{eq:stationary}
\nabla_{\vec{t}_{l,k}} \E\left[c_k(\rvec{H},\rvec{t}_{-l,k}^\star,\vec{t}_{l,k})\middle| S_l \right]\bigg|_{\vec{t}_{l,k} = \rvec{t}_{l,k}^\star(S_l)} = \vec{0} \quad  \mathrm{a.s.}, \quad \forall l \in \set{L},
\end{equation}
where $(\rvec{t}_{-l,k},\vec{t}_{l,k}) := (\rvec{t}_{1,k},\ldots,\rvec{t}_{l-1,k},\vec{t}_{l,k},\rvec{t}_{l+1,k},\ldots,\rvec{t}_{L,k})$.
\end{definition}
By evaluating the stationary conditions \eqref{eq:stationary} using standard results on differentiation of real-valued quadratic forms over a complex domain, we obtain that a stationary solution may be given by any solution to the following feasibility problem:
\begin{equation}\label{eq:stationary_evaluated}
\text{find } \rvec{t}_k\in \set{T} \text{ s.t. } \E[\rvec{Q}_{l,l}|S_l]\rvec{t}^\star_{l,k}(S_l) + \sum_{j\neq l}\E[\rvec{Q}_{l,j}\rvec{t}^\star_{j,k}|S_l] - \E[\rvec{g}_{l,k}|S_l] = \vec{0} \quad \text{a.s.}, \quad \forall l \in \set{L},
\end{equation}
where $\rvec{Q}_{l,l} := \rvec{H}_l^\herm\rvec{H}_l+\frac{1}{P}\vec{I}$, $\rvec{Q}_{l,j} := \rvec{H}_l^\herm\rvec{H}_j$ for $j\neq l$, and $\rvec{g}_{l,k} := \rvec{H}_l^\herm\vec{e}_k$, provided that all expectations are finite.
Since the considered quadratic cost $c_k$ is convex and differentiable a.s. in each of the $\vec{t}_{l,k}$, under some mild technical assumptions the notion of stationarity can be interpreted as enforcing each function $\rvec{t}^\star_{l,k}(S_l)$ to be optimal while keeping the functions $\rvec{t}^\star_{j,k}(S_j)$ of all the other TXs $j \neq l$ fixed. This is reminiscent of the game theoretical notion of \textit{Nash equilibrium}, with the difference that here all the TXs share the same objective, and hence they act as a team. Similarly to Nash equilibria, stationary solutions may be in general inefficient, i.e., lead to a local optimum. However, a stronger result holds for quadratic teams:
\begin{theorem}
\label{th:quadratic_teams}
If $\rvec{Q}$ is uniformly bounded above, i.e., there exists a positive scalar $B<\infty$ such that $\rvec{Q}\prec B\vec{I}$ a.s., then Problem \eqref{eq:team_prob} admits a unique optimal solution, which is also the unique stationary solution solving Problem \eqref{eq:stationary_evaluated}.
\end{theorem}
\begin{proof}
Theorem~2.6.6 of \cite{yukselbook}. 
\end{proof}
Theorem~\ref{th:quadratic_teams} and Problem \eqref{eq:stationary_evaluated} are of fundamental theoretical importance since they concisely identify the two key ingredients for optimal distributed precoding design:
\begin{enumerate}
\item Robustness against local channel estimation errors, captured by $\E[\rvec{Q}_{l,l}|S_l]$ and $\E[\rvec{g}_{l,k}|S_l]$;
\item Robustness against the effect of the ``decisions" taken at the other TXs, captured by $\sum_{j\neq l}\E[\rvec{Q}_{l,j}\rvec{t}^\star_{j,k}|S_l]$. This is the main new difficulty which is introduced while moving from centralized to distributed precoding design.
\end{enumerate}
From a practical perspective, the above results also provide a very powerful tool to solve the difficult distributed precoding design problem. Specifically, they provide a set of \textit{optimality conditions} for Problem~\eqref{eq:team_prob} in the form of a standard infinite dimensional \textit{linear} feasibility problem, for which many approximate solution methods are available. For instance, the optimal TMMSE precoders may be approached via one of the iterative methods surveyed in \cite{yukselbook} based on interpreting the solution to  \eqref{eq:stationary_evaluated} as the unique fixed point of a linear map. Other promising methods may also include finite dimensional approximations of  \eqref{eq:stationary_evaluated} obtained, e.g., by sampling the CSI process $(\rvec{H},S_1,\ldots,S_L)$ and by interpreting the sampled version of \eqref{eq:stationary_evaluated} as a classical function interpolation problem from a finite set of linear measurements \cite{wahba1990spline}. Further discussions on approximate solution methods are left for future work, and most parts of this study will focus on cell-free massive MIMO networks and in particular on special cases where \eqref{eq:stationary_evaluated} can be solved explicitly. However, we remark that the content of this section can be readily applied to study general networks with distributed CSIT as described in Section~\ref{sec:model}, not necessarily restricted to cell-free massive MIMO networks.

Before moving to the aforementioned results, we focus on a rather technical yet important weakness of Theorem~\ref{th:quadratic_teams}. The assumption of Theorem~\ref{th:quadratic_teams} is essentially used to ensure the existence of all the expectations in the steps of the proof, and is satisfied for any fading distribution with bounded support. However, it is not satisfied for the classical Gaussian fading model. Despite being unrealistic, since physically consistent fading distributions cannot have unbounded support, Gaussian fading is a very common model in the literature due to its analytical treatability, for example in deriving simple channel estimation error models \cite{tse2005fundamentals,massivemimobook}. Furthermore, except for the tails of the distribution, it usually fits measurements well. Therefore, in the following we derive more general optimality results covering this case. 
\begin{theorem}
\label{th:quadratic_teams_new}
If $\E[\|\rvec{Q}\|^2_{\mathrm{F}}]<\infty$,
then Problem \eqref{eq:team_prob} admits a unique optimal solution, which is also the unique stationary solution solving Problem \eqref{eq:stationary_evaluated}.
\end{theorem}
\begin{proof} The proof is given in Appendix \ref{proof:quadratic_teams_new}. 
\end{proof}
Finally, we conclude this section by observing that the optimality conditions \eqref{eq:stationary_evaluated} are not only useful to characterize the optimal TMMSE solution, but also to evaluate the suboptimality of its approximations. This can be done via an appropriate measure of violation of the optimality conditions. Specifically, we have the following result:
\begin{lemma}\label{lem:bound}
Suppose that $\E[\|\rvec{Q}\|^2_{\mathrm{F}}]<\infty$ holds, and let $\rvec{t}_k^\star \in \set{T}$ be the unique solution to Problem \eqref{eq:team_prob}. Furthermore, define $\rvec{z}_k := [\rvec{z}_{1,k}^\T,\dots,\rvec{z}^\T_{L,k}]^\T$, where $\rvec{z}_{l,k} = \rvec{z}_{l,k}(S_l)$ is given by the left-hand side of the stationary conditions in \eqref{eq:stationary_evaluated} with $\rvec{t}_k^\star$ replaced by an arbitrary $\rvec{t}_k \in \set{T}$.
If $\rvec{z}_k \in \set{T}$, i.e., if $\E[\|\rvec{z}_{l,k}(S_l)\|^2]<\infty$, the following optimality bounds hold:
\begin{align}
\mathrm{MSE}_k(\rvec{t}_k)-\mathrm{MSE}_k(\rvec{t}_k^\star) &\leq \E\left[\|\rvec{Q}^{-\frac{1}{2}}\rvec{z}_k\|^2\right] \label{eq:first_bound} \\
&\leq P\mathbb{E}\left[\|\rvec{z}_k\|^2\right] \label{eq:bound},
\end{align}
\end{lemma}
\begin{proof}
The proof is given in Appendix \ref{proof:bound}. 
\end{proof}  
Clearly, $\rvec{z}_k(S_1,\ldots,S_L) = \vec{0}$ a.s. gives the optimality conditions in \eqref{eq:stationary_evaluated}, and in fact it corresponds to a zero optimality gap in \eqref{eq:bound}. Intuitively, the bounds in \eqref{eq:bound} can be quite tight if $\rvec{z}_k(S_1,\ldots,S_L) \approx \vec{0}$ with high probability. However, if this is not satisfied, we remark that both bounds can be looser than other trivial bounds obtained, e.g., by assuming a centralized information structure, or even output negative estimates of $\mathrm{MSE}_k(\rvec{t}_k^\star)$. As already mentioned, we leave further studies on suboptimal solutions for future work, and we use \eqref{eq:bound} only in Section \ref{ssec:iid_fading} for getting analytical insights into a particular setup.

\section{Applications to cell-free massive MIMO}
\label{sec:applications}
In this section we specialize the theory of Section \ref{sec:TMMSE} to cell-free massive MIMO networks, and explicitly derive optimal TMMSE precoders for two practical examples. In the scope of this study, the important feature of the cell-free massive MIMO paradigm is the exploitation of time division duplex operations and channel reciprocity to efficiently acquire estimates $\hat{\rvec{H}}_l$ of the local channel $\rvec{H}_l$ at each TX~$l$ via over-the-uplink training \cite{ngo2017cellfree}. 
These estimates may be subsequently shared through the fronthaul according to some predefined CSIT sharing mechanism, forming at each TX~$l$ a side information about the global channel $\rvec{H}$ of the type
\begin{equation}
S_l:= \left(\hat{\rvec{H}}_l,\bar{S}_l\right)
\end{equation}
where $\bar{S}_l$ denotes the side information about the other channels $\{\rvec{H}_j\}_{j\neq l}$ collected at TX~$l$. Depending on the CSIT sharing mechanism, $\bar{S}_l$ may be a function of the other local channel estimates $\{\hat{\rvec{H}}_j\}_{j\neq l}$ (e.g., in case of error-free digital signalling), or include additional noise (e.g., in case of  random events such as protocol delays). Consistently with the above discussion, we consider the following assumptions: 
\begin{assumption}[Local channel estimation]\label{ass:TDD}
For every $l\in \set{L}$, let $\rvec{E}_l:= \rvec{H}_l - \hat{\rvec{H}}_l$ be the local channel estimation error for the local channel. Assume that $\hat{\rvec{H}}_l$ and $\rvec{E}_l$ are independent. Furthermore, assume $\E[\rvec{E}_l] = \vec{0}$, and that $\E[\rvec{E}_l^\herm\rvec{E}_l] =: \vec{\Sigma}_l$ has finite elements. Finally, assume that $(\hat{\rvec{H}}_l,\rvec{E}_l)$ and $(\hat{\rvec{H}}_j,\rvec{E}_j)$ are independent for $l\neq j$.
\end{assumption} 

\begin{assumption}[CSIT sharing mechanism]\label{ass:CSIT_sharing}
For every $(l,j) \in \set{L}^2$ s.t. $l\neq j$, assume the following Markov chain: $
\rvec{H}_l \to \hat{\rvec{H}}_l \to S_l \to S_j \to \hat{\rvec{H}}_j \to \rvec{H}_j$.
\end{assumption}
Assumption \ref{ass:TDD} is widely used in the wireless communication literature and it holds, e.g., for pilot-based MMSE estimates of Gaussian channels \cite{ngo2017cellfree,emil2020scalable}.  Assumption \ref{ass:CSIT_sharing} essentially states that all the  available information about $\rvec{H}_l$ is fully contained in $\hat{\rvec{H}}_l$ at TX $l$, and that TX $j$ can only obtain a degraded version of it. We now rewrite the optimality conditions given by \eqref{eq:stationary_evaluated} in light of the considered model: 

\begin{lemma}\label{lem:stationarity_imperfect}
Suppose that Assumption~\ref{ass:TDD}, Assumption~\ref{ass:CSIT_sharing}, and the assumption of Theorem~\ref{th:quadratic_teams_new} hold. Then, the unique TMMSE solution to Problem~\eqref{eq:team_prob} is given by the unique $\rvec{t}_k^\star \in \set{T}$ satisfying
\begin{equation}\label{eq:stationary_imperfect}
\rvec{t}_{l,k}^\star(S_l) =  \rvec{F}_l\left(\vec{e}_k-\sum_{j\neq l} \E\left[\hat{\rvec{H}}_j\rvec{t}^\star_{j,k}\Big|S_l\right] \right) \quad \mathrm{a.s.}, \quad \forall l \in \set{L},
\end{equation} 
where $\rvec{F}_l:=\left(\hat{\rvec{H}}_l^\herm\hat{\rvec{H}}_l+\vec{\Sigma}_l + P^{-1}\vec{I}\right)^{-1}\hat{\rvec{H}}_l^\herm$.
\end{lemma}
\begin{proof} 
The first term of the stationarity conditions \eqref{eq:stationary_evaluated} is evaluated by letting
\begin{equation}
\begin{split}
\E[\rvec{Q}_{l,l}|S_l] &= \E[\rvec{H}_l^\herm\rvec{H}_l|S_l]+P^{-1}\vec{I}\\
&= \E[(\hat{\rvec{H}}_l+\rvec{E}_l)^\herm(\hat{\rvec{H}}_l+\rvec{E}_l)|\hat{\rvec{H}}_l]+P^{-1}\vec{I}\\
&=\hat{\rvec{H}}_l^\herm\hat{\rvec{H}}_l +\vec{\Sigma}_l+P^{-1}\vec{I},
\end{split}
\end{equation} 
where we used the Markov chain $\rvec{H}_l \to \hat{\rvec{H}}_l\to S_l$ and Assumption \ref{ass:TDD}. Then, for $j\neq l$:
\begin{equation}
\begin{split}
\E[\rvec{Q}_{l,j}\rvec{t}_j^\star|S_l] &= \E[\rvec{H}_l^\herm\rvec{H}_j\rvec{t}_{j,k}^\star|S_l]\\
&\overset{(a)}{=} \E\left[\E[\rvec{H}_l^\herm\rvec{H}_j|S_l,S_j]\rvec{t}_{j,k}^\star\middle|S_l \right]\\
&\overset{(b)}{=} \E\left[\E[\rvec{H}_l^\herm|S_l,S_j]\E[\rvec{H}_j|S_l,S_j]\rvec{t}_{j,k}^\star\middle|S_l \right]\\
&\overset{(c)}{=} \E\left[\E[\rvec{H}_l^\herm|\hat{\rvec{H}}_l]\E[\rvec{H}_j|\hat{\rvec{H}}_j]\rvec{t}_{j,k}^\star\middle|S_l \right]\\
&\overset{(d)}{=} \hat{\rvec{H}}_l^\herm\E[\hat{\rvec{H}}_j\rvec{t}_{j,k}^\star|S_l],
\end{split}
\end{equation}
where $(a)$ follows from the law of total expectation and $\rvec{t}_{j,k}=\rvec{t}_{j,k}(S_j)$, $(b)$ from the Markov chain $\rvec{H}_l \to (S_l,S_j) \to \rvec{H}_j$, $(c)$ from the Markov chain $\rvec{H}_l \to \hat{\rvec{H}}_l\to (S_l,S_j)$, and $(d)$ from Assumption~\ref{ass:TDD}. Note that all the aforementioned Markov chains are implied by Assumption~\ref{ass:CSIT_sharing}. The proof is concluded by using $
\E[\rvec{g}_{l,k}|S_l] = \hat{\rvec{H}}_l^\herm\vec{e}_k$ and by rearranging the terms.
\end{proof}
The above lemma reveals the following structure of the optimal TMMSE solution: the matrix $\rvec{F}_l$ can be recognized as a \textit{local} MMSE precoding stage (studied, e.g., in \cite{emil2020scalable}), that is, a centralized MMSE solution \cite{massivemimobook} assuming that there are no other TXs than TX $l$; the remaining part can be then interpreted as a `corrective' stage which takes into account the effect of the other TXs based on the available CSIT and long-term statistical information.

\subsection{No CSIT sharing}
\label{ssec:localCSIT}
As an important example, we assume that no local channel estimate is shared along the fronthaul. This corresponds to the original cell-free massive MIMO setup studied in \cite{ngo2017cellfree}. Specifically, we let $\hat{\rvec{H}}_l$ as in Assumption~\ref{ass:TDD} and 
\begin{equation}\label{eq:local_CSIT}
S_l = \hat{\rvec{H}}_l, \quad \forall l\in \set{L}.
\end{equation}
\begin{theorem}
\label{th:teamMMSE_local}
The TMMSE precoders solving \eqref{eq:stationary_imperfect} under no CSIT sharing \eqref{eq:local_CSIT} are given by 
\begin{equation}
\label{eq:localTMMSE}
\rvec{t}_{l,k}^\star(S_l)=\rvec{F}_l\vec{C}_l\vec{e}_k, \quad \forall l \in \set{L},
\end{equation}
for some matrices of coefficients $\vec{C}_l \in \mathbb{
C}^{K\times K}$. Furthermore, the optimal $\vec{C}_l$ are given by the unique solution of the linear system
$\vec{C}_l + \sum_{j\neq l}\vec{\Pi}_j\vec{C}_j = \vec{I}$, $\forall l \in \set{L}$, where $\vec{\Pi}_l := \E\left[\hat{\rvec{H}}_l\rvec{F}_l\right]$.
\end{theorem}
\begin{proof}
Substituting \eqref{eq:localTMMSE} into the optimality conditions \eqref{eq:stationary_imperfect}, we need to show that
\begin{equation}
\hat{\rvec{H}}_l^\herm \left(\vec{C}_l + \sum_{j\neq l}\mathbb{E}\left[\hat{\rvec{H}}_j\rvec{F}_j\vec{C}_j\middle| S_l\right]-\vec{I}\right)\vec{e}_k = \vec{0} \quad \mathrm{a.s. }, \quad \forall l \in \set{L}.
\end{equation}
By the independence between $\hat{\rvec{H}}_l$ and $\hat{\rvec{H}}_j$, we can drop the conditioning on $S_l$ and obtain $\hat{\rvec{H}}_l^\herm \left(\vec{C}_l + \sum_{j\neq l}\vec{\Pi}_j\vec{C}_j-\vec{I}\right)\vec{e}_k = \vec{0}$ a.s., $\forall l \in \set{L}.$ The proof is concluded by observing that $
\vec{C}_l + \sum_{j\neq l}\vec{\Pi}_j\vec{C}_j=\vec{I}$, $\forall l \in \set{L}$, always has a unique solution, as shown in Appendix~\ref{proof:teamMMSE_local}. 
\end{proof}
The optimal solution \eqref{eq:localTMMSE} corresponds to a two-stage precoding scheme composed by a local MMSE precoding stage $\rvec{F}_l$ preceded by a statistical precoding stage $\vec{C}_l$. By letting the rows $\hat{\rvec{g}}_{l,k}^\herm$ of $\hat{\rvec{H}}_l$ to be independent and distributed as $\CN(\vec{0},\vec{K}_{l,k})$, corresponding for instance to a non-line-of-sight (NLoS) scenario with no pilot contamination \cite{emil2020scalable}, it can be shown that the matrices $\vec{\Pi}_l$ are diagonal. Hence, \eqref{eq:localTMMSE} takes the simpler form 
\begin{equation}\label{eq:localMMSE}
\rvec{t}_{l,k}^\star(S_l)=c_{l,k}\rvec{F}_l\vec{e}_k,
\end{equation} 
which, by mapping the optimal $c_{l,k}$ to the optimal large-scale fading decoding coefficients in a dual UL channel, was already studied in \cite{emil2020cellfree}. However, if the channels have non-zero mean, such as in line-of-sight (LoS) models, \eqref{eq:localTMMSE} may provide significantly higher rates than \eqref{eq:localMMSE}. To see this, let $\hat{\rvec{H}}_l \approx \bar{\vec{H}}_l$ for some fixed matrix $\bar{\vec{H}}_l$, $\forall l\in \set{L}$. Then, since $\bar{\vec{H}}_l$ is statistical information known to all TXs, the TMMSE precoders should take a form similar to a `long-term' centralized MMSE solution, which cannot be implemented using \eqref{eq:localMMSE}. Finally, we point out that a suboptimal variation of \eqref{eq:localTMMSE} called \textit{optimal bilinear equalizer} (OBE), with $\rvec{F}_l$ replaced by $\hat{\rvec{H}}_l^\herm$, was already proposed in \cite{neumann2018bilinear} as a low-complexity alternative to centralized MMSE precoding which maintains robustness against pilot contamination.

\subsection{Unidirectional CSIT sharing}
\label{ssec:unidirectional}
We now consider a more involved example and let the local channel measurements be shared unidirectionally along a serial fronthaul. This setup is relevant, e.g., for the cell-free massive MIMO network in Figure \ref{fig:stripe}, where CSIT, messages, and power are distributed along a serial fronthaul from and/or towards a CPU located at one edge, an architecture also known as a \textit{radio stripe} \cite{interdonato2019ubiquitous,shaik2020mmse}.  
\begin{figure}[ht!]
\centering
\begin{overpic}[width=0.8\columnwidth,tics=5]{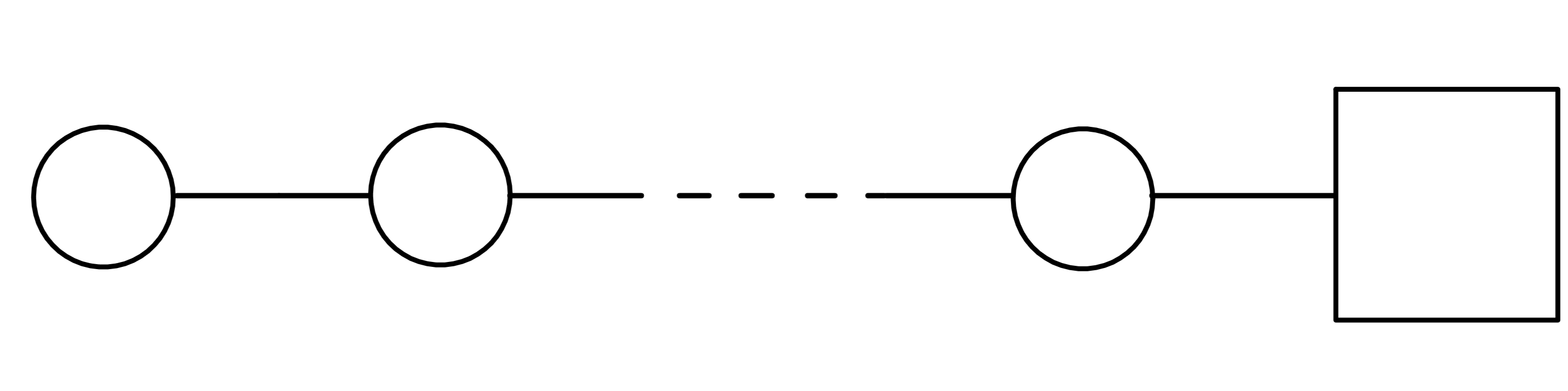}
 \put(3.8,11.5){\small TX$\,1$}
 \put(25.3,11.5){\small TX$\,2$}
 \put(66.1,11.5){\small TX$\,L$}
 \put(89.2,11.5){CPU}
 
 \put(1,20){\small (a) \normalsize $\hat{\rvec{H}}_1$}
 \put(26,20){$\hat{\rvec{H}}_2$}
 \put(67,20){$\hat{\rvec{H}}_L$}
 
 \put(1,2){\small (b) \normalsize $\hat{\rvec{H}}_1$}
 \put(13,2){$\longrightarrow$}
 \put(23.5,2){$\hat{\rvec{H}}_1,\hat{\rvec{H}}_2$}
 \put(37.5,2){$\longrightarrow$}
 \put(58,2){$\hat{\rvec{H}}_1,\hat{\rvec{H}}_2,\ldots,\hat{\rvec{H}}_L$}
\end{overpic} 
\caption{Pictorial representation of a radio stripe with (a) no CSIT sharing, and (b) unidirectional CSIT sharing.}
\label{fig:stripe}
\end{figure}
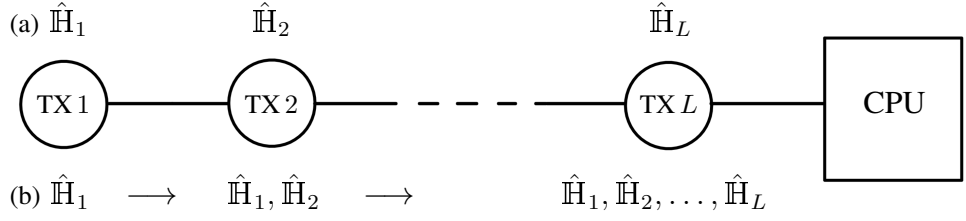

Specifically, $\forall l\in \set{L}$, we let $\hat{\rvec{H}}_l$ as in Assumption~\ref{ass:TDD} and
\begin{equation}\label{eq:unidirectional_CSIT}
S_l = (\hat{\rvec{H}}_1,\ldots,\hat{\rvec{H}}_l).
\end{equation}
This particular information structure can be interpreted as the CSIT which is accumulated at every TX during the first phase of a centralized precoding scheme for radio stripes, where the CPU collects the $K\times LN$ channel matrix $\vec{H}$ through the serial fronthaul.

\begin{theorem}\label{th:teamMMSE_undirectional}
The TMMSE precoders solving \eqref{eq:stationary_imperfect} under unidirectional CSIT sharing \eqref{eq:unidirectional_CSIT} are given by 
\begin{equation}\label{eq:teamMMSE_solution}
\rvec{t}_{l,k}^\star(S_l)=\rvec{F}_l \rvec{V}_l\left[ \prod_{i=1}^{l-1}\bar{\rvec{V}}_i\right]\vec{e}_k, \quad \forall l \in \set{L},
\end{equation}
where we use the following short-hands:
\begin{itemize}
\item $\rvec{V}_l:=\left(\vec{I}-\vec{\Pi}_l\rvec{P}_l\right)^{-1}(\vec{I}-\vec{\Pi}_l)$;
\item $\bar{\rvec{V}}_l:=\vec{I}-\rvec{P}_l\rvec{V}_l$;
\item $\rvec{P}_l := \hat{\rvec{H}}_l\rvec{F}_l$;
\item $\vec{\Pi}_l:=
 \E[\rvec{P}_{l+1} \rvec{V}_{l+1}]+\vec{\Pi}_{l+1}\E[\bar{\rvec{V}}_{l+1}]$, $\vec{\Pi}_L:=\vec{0}$.
\end{itemize}
\end{theorem}
\begin{proof}
We first assume that all the matrix inverses involved in the following steps exist. Substituting \eqref{eq:teamMMSE_solution} into  \eqref{eq:stationary_imperfect}, we need to show that 
\begin{equation}
\hat{\rvec{H}}_l^\herm \left(\rvec{V}_l\prod_{i=1}^{l-1}\bar{\rvec{V}}_i + \sum_{j\neq l}\mathbb{E}\left[\rvec{P}_j\rvec{V}_j\prod_{i=1}^{j-1}\bar{\rvec{V}}_i\middle| S_l\right]-\vec{I}\right)\vec{e}_k = \vec{0} \quad \mathrm{a.s.}, \quad \forall l \in \set{L}.
\end{equation}
To verify the above statement, we rewrite the first two terms inside the outer brackets as:
\begin{equation}\label{eq:stationarity_stripes}
\left(\rvec{V}_l + \sum_{j> l}\mathbb{E}\left[\rvec{P}_j\rvec{V}_j\prod_{i=l+1}^{j-1}\bar{\rvec{V}}_i \right]\bar{\rvec{V}}_l\right)\prod_{i=1}^{l-1}\bar{\rvec{V}}_i +\sum_{j< l}\rvec{P}_j\rvec{V}_j\prod_{i=1}^{j-1}\bar{\rvec{V}}_i,
\end{equation}
where we use the fact that $\rvec{P}_j$, $\rvec{V}_j$, and $\bar{\rvec{V}}_j$ are deterministic functions of $\hat{\rvec{H}}_j$ only, hence they are independent from $S_l$ for $j>l$, while they are deterministic functions of $S_l$ otherwise. Furthermore, since $\rvec{P}_j$, $\rvec{V}_j$, and $\bar{\rvec{V}}_j$ are independent from $\rvec{P}_i$, $\rvec{V}_i$, and $\bar{\rvec{V}}_i$ $\forall i \neq j$, we have
\begin{equation}
\begin{split}
\sum_{j> l}\mathbb{E}\left[\rvec{P}_j\rvec{V}_j\prod_{i=l+1}^{j-1}\bar{\rvec{V}}_i \right] &= \sum_{j> l}\mathbb{E}\left[\rvec{P}_j\rvec{V}_j\right]\prod_{i=l+1}^{j-1}\mathbb{E}\left[\bar{\rvec{V}}_i \right]\\
&= \mathbb{E}\left[\rvec{P}_{l+1}\rvec{V}_{l+1}\right] + \sum_{j >l+1}\mathbb{E}\left[\rvec{P}_j\rvec{V}_j\right]\prod_{i=l+1}^{j-1}\mathbb{E}\left[\bar{\rvec{V}}_i \right] \\
&= \mathbb{E}\left[\rvec{P}_{l+1}\rvec{V}_{l+1}\right] + \left(\sum_{j >l+1}\mathbb{E}\left[\rvec{P}_j\rvec{V}_j\right]\prod_{i=l+2}^{j-1}\mathbb{E}\left[\bar{\rvec{V}}_i\right]\right)\mathbb{E}\left[\bar{\rvec{V}}_{l+1} \right].
\end{split}
\end{equation}
The second and last term of the above chain of equalities define a recursion terminating with 
$\mathbb{E}\left[\rvec{P}_L\rvec{V}_L\right] + \vec{0} \mathbb{E}\left[\bar{\rvec{V}}_L\right]= \vec{\Pi}_{L-1}$. This recursion gives precisely $\sum_{j> l}\mathbb{E}\left[\rvec{P}_j\rvec{V}_j\prod_{i=1}^{j-1}\bar{\rvec{V}}_i \right] = \vec{\Pi}_l$. Together with the property  $\rvec{V}_l + \vec{\Pi}_l\bar{\rvec{V}}_l = \vec{I}$, \eqref{eq:stationarity_stripes} simplifies to 
\begin{equation}
\begin{split}
\prod_{i=1}^{l-1}\bar{\rvec{V}}_i +\sum_{j< l}\rvec{P}_j\rvec{V}_j\prod_{i=1}^{j-1}\bar{\rvec{V}}_i &= \left(\bar{\rvec{V}}_{l-1} + \rvec{P}_{l-1}\rvec{V}_{l-1}\right)\prod_{i=1}^{l-2}\bar{\rvec{V}}_i+\sum_{j <l-1}\rvec{P}_j\rvec{V}_j\prod_{i=1}^{j-1}\bar{\rvec{V}}_i\\ 
&= \prod_{i=1}^{l-2}\bar{\rvec{V}}_i +\sum_{j <l-1}\rvec{P}_j\rvec{V}_j\prod_{i=1}^{j-1}\bar{\rvec{V}}_i,
\end{split}
\end{equation}
where the last equation follows from the definition of $\bar{\rvec{V}}_l$, and where we identify another recursive structure among the remaining terms. By continuing until termination, we finally obtain $\prod_{i=1}^{l-1}\bar{\rvec{V}}_i +\sum_{j< l}\rvec{P}_j\rvec{V}_j\prod_{i=1}^{j-1}\bar{\rvec{V}}_i = \vec{I}$, which proves the main statement under the assumption that all the matrix inverses involved exist. This assumption is indeed always satisfied, as shown in Appendix \ref{proof:teamMMSE_unidirectional}.
\end{proof} 
By locally computing precoders based on $S_l$ only, and at the expense of some performance loss, the scheme in \eqref{eq:teamMMSE_solution} eliminates the additional overhead required by centralized precoding to share back the computed $K\times LN$ precoding matrix from the CPU to the TXs. Furthermore, inspired by the schemes proposed in \cite{interdonato2019ubiquitous,emil2020cellfree,shaik2020mmse} for UL processing exploiting the peculiarity of a serial fronthaul, the CSIT sharing overhead can be further reduced as follows:
\begin{remark}\label{rem:recursive}
The scheme in \eqref{eq:teamMMSE_solution} can be alternatively implemented via a recursive algorithm involving a $K\times K$ aggregate information  matrix $\prod_{i=1}^{l-1}\bar{\rvec{V}}_i$ which is sequentially processed and forwarded in the direction from TX $1$ to TX $L$. Therefore, the capacity of the serial fronthaul can be made independent from $L$, which is typically larger than $K$.

Furthermore, if data sharing is implemented through the sequential forwarding of a vector $\rvec{u}:=[U_1,\ldots,U_K]^\T \in \stdset{C}^K$ of coded and modulated I/Q symbols originating from a CPU placed next to TX $1$, then this can be replaced by the forwarding of a sequentially precoded $K$-dimensional vector $\prod_{i=1}^{l-1}\bar{\rvec{V}}_i\rvec{u}$, thus eliminating the CSIT sharing overhead.
\end{remark}

We conclude this section by providing the following corollary to Theorem \ref{th:teamMMSE_undirectional}.
\begin{corollary}\label{cor:bidirectional_TMMSE}
An alternative expression for centralized MMSE precoding \cite{massivemimobook,emil2020scalable}, or equivalently, for the TMMSE solution under full CSIT sharing $S_l = (\hat{\rvec{H}}_1,\ldots,\hat{\rvec{H}}_L)$ $\forall l \in \set{L}$, is given by \eqref{eq:teamMMSE_solution} with $\vec{\Pi}_l$ replaced by $
\bar{\rvec{P}}_l:=
 \rvec{P}_{l+1} \rvec{V}_{l+1}+\bar{\rvec{P}}_{l+1}\bar{\rvec{V}}_{l+1}$, $\bar{\rvec{P}}_L:=\vec{0}$.
\end{corollary}
\begin{proof}
Since all random quantities become deterministic after conditioning on $S_l$, the proof of Theorem \ref{th:teamMMSE_undirectional} can be repeated by removing $\E[\cdot]$ everywhere.
\end{proof}
The expression in Corollary \ref{cor:bidirectional_TMMSE} can be alternatively derived by applying recursively known block-matrix inversion lemmas to the original centralized MMSE precoding expression \cite{massivemimobook}. The details are omitted due to space limitations. Similarly to the implementation of \eqref{eq:teamMMSE_solution} described in Remark \ref{rem:recursive}, Corollary \ref{cor:bidirectional_TMMSE} provides a novel distributed and recursive implementation of centralized MMSE precoding. The main difference is that, in contrast to $\vec{\Pi}_l$ which can be computed offline, the computation of $\bar{\rvec{P}}_l$ entails an additional sequential procedure in the reverse direction, thus increasing the overhead. 

\subsection{Asymptotic results and relation with the SGD scheme \cite{rodriguez2020decentralized}}
\label{ssec:iid_fading}
The idea of designing recursive precoding schemes exploiting the opportunities of a serial connection between antenna elements has been also explored by \cite{rodriguez2020decentralized}. Motivated by the need of reducing hardware complexity of a massive MIMO cellular base station, and by focusing on $N=1$ and no channel estimation error, the authors of \cite{rodriguez2020decentralized} propose the following so-called \textit{SGD} precoding scheme: 
\begin{equation}\label{eq:SGD}
T_{l,k}(S_l) = \mu_{l,k}\rvec{h}_l^\herm \left(\vec{e}_k-\sum_{j=1}^{l-1} \rvec{h}_j T_{j,k}(S_j) \right), \quad \forall l \in \set{L},
\end{equation}
where $\rvec{h}_l:= \hat{\rvec{H}}_l = \rvec{H}_l$, $S_l$ is given by \eqref{eq:unidirectional_CSIT} assuming unidirectional CSIT sharing, and $\mu_{l,k} \in \stdset{R}$ are tunable step-sizes of a stochastic gradient descent algorithm. The choice $\mu_{l,k} = \|\rvec{h}_l\|^{-2}$ is motivated by \cite{rodriguez2020decentralized} as a good solution for i.i.d. Rayleigh fading and high SNR. Furthermore, to cope with finite SNR, \cite{rodriguez2020decentralized} suggests to take $\mu_{l,k} = \mu_k\|\rvec{h}_l\|^{-2}$ for a single deterministic scalar $\mu_k\in\stdset{R}$ per RX to be optimized, e.g., using line search. Interestingly, the SGD scheme with $\mu_{l,k} = \|\rvec{h}_l\|^{-2}$ can be also derived from team theoretical arguments, as a particular case of the following asymptotic result considering $N\geq 1$:
\begin{lemma}\label{lem:bound_high_SNR}
Assume $\mathrm{vec}(\rvec{H})\sim \CN(\vec{0},\vec{I})$, $\hat{\rvec{H}}_l = \rvec{H}_l$ $\forall l\in \set{L}$, unidirectional CSIT sharing \eqref{eq:unidirectional_CSIT}, and let $\rvec{t}_k^\star$ be the optimal TMMSE solution of Problem \eqref{eq:team_prob}. Then,
\begin{equation}\label{eq:bound_high_SNR}
R_k = \log(\mathrm{MSE}_k(\rvec{t}_k^\star))^{-1} \leq L \log\left(\dfrac{K}{K-N} \right),
\end{equation}
with equality attained as $P\to \infty$ by 
\begin{equation}\label{eq:sequentialZF}
\rvec{t}_{l,k}(S_l) = (\rvec{H}_l^\herm\rvec{H}_l)^{-1}\rvec{H}_l^\herm \left(\vec{e}_k-\sum_{j=1}^{l-1} \rvec{H}_j \rvec{t}_{j,k}(S_j) \right), \quad \forall l \in \set{L}.
\end{equation}
\end{lemma}
\begin{proof}
The proof is given in Appendix \ref{proof:bound_high_SNR}.
\end{proof}

\section{Performance evaluation}
\label{sec:simulations}
\subsection{Simulation setup}
Inspired by the ``football arena" \cite{shaik2020mmse} or ``outdoor piazza" \cite{interdonato2019ubiquitous} scenarios, we simulate a network with a radio stripe of $L=30$ equally spaced TXs with $N=2$ antennas each wrapped around a circular area of radius $r_1 = 60$ m, and $K=7$ RXs independently and uniformly drawn within a concentric circular  area of radius $r_2= 50$ m. We  let the channel coefficient $h_{l,k,n}$ between the $n$-th antenna of TX $l$ and RX $k$ be independently distributed as $H_{l,k,n} \sim \CN(0,\rho^2_{l,k})$, where $\rho^2_{l,k}$ denotes the channel gain between TX $l$ and RX $k$. We follow the 3GPP NLoS Urban Microcell path-loss model \cite[Table B.1.2.1-1]{3GPP}
\begin{equation}
\mathrm{PL}_{l,k} = 36.7 \log_{10}\left(\dfrac{d_{l,k}}{1 \; \mathrm{m}}\right) + 22.7 + 26\log_{10}\left(\dfrac{f_c}{1 \; \mathrm{GHz}}\right) \quad [\text{dB}],
\end{equation}
where $f_c = 2$ GHz is the carrier frequency, and $d_{l,k}$ is the distance between TX~$l$ and RX~$k$ including a difference in height of $10$ m. We let the noise power at all RXs be given by
$
P_{\mathrm{noise}} = -174 + 10 \log_{10}(B/1 \; \mathrm{Hz}) + F$ dBm,
where $B = 20$ MHz is the system bandwidth, and $F = 7$ dB is the noise figure. Finally, we let $\rho^2_{l,k} := 10^{-\frac{\mathrm{PL}_{l,k} +P_{\mathrm{noise}}}{10}}$ $\text{mW}^{-1}$, and, leveraging the short distances, we consider a relatively low total radiated power $P_{\mathrm{sum}} = 100$ mW. 

\begin{figure}[th!]
\centering
\includegraphics[width=0.5\columnwidth]{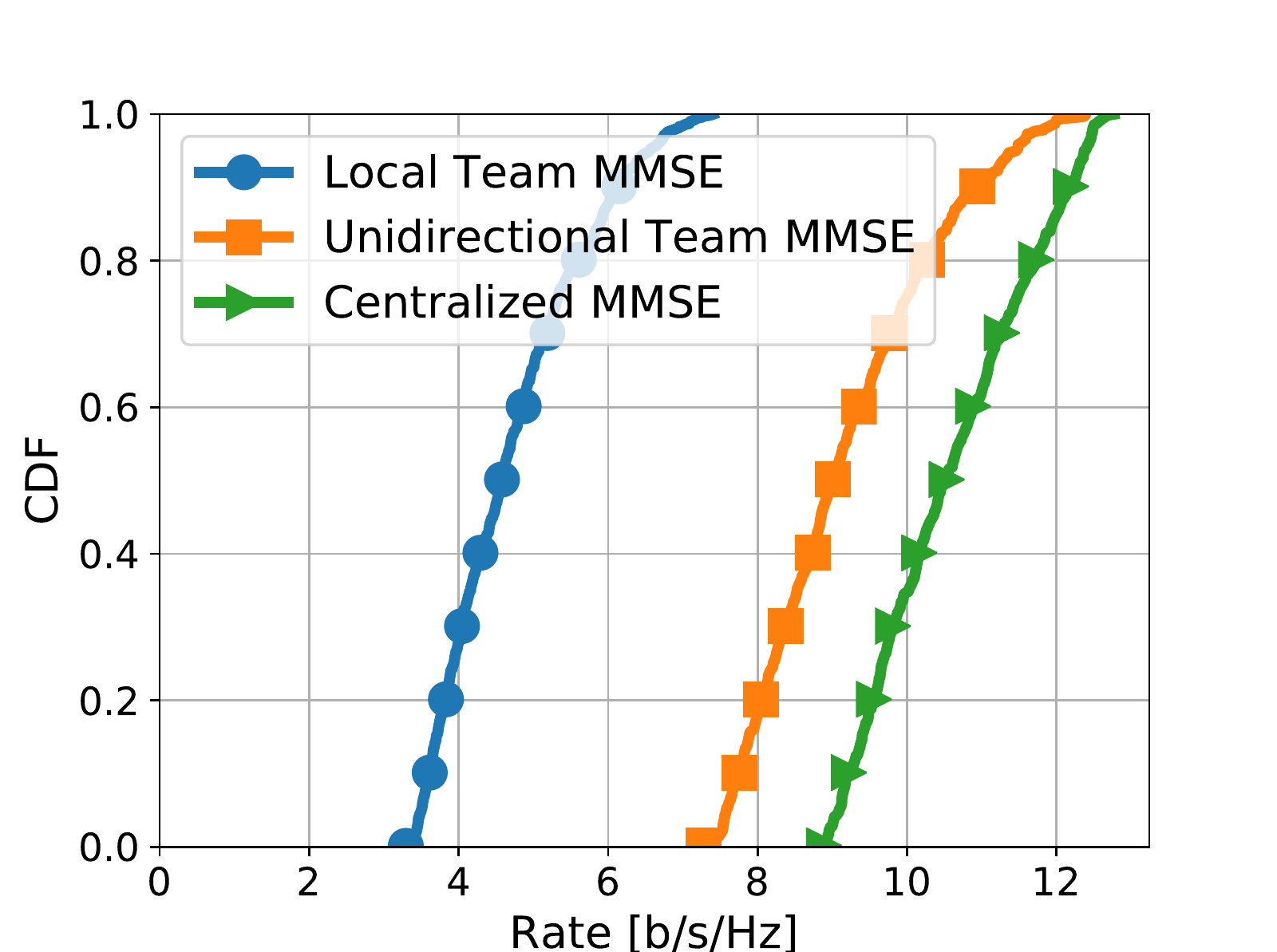}
\label{fig:cdf_stadium}
\caption{Comparison among different CSIT configuration, empirical CDF of the optimal per-RX achievable rates. Unidirectional TMMSE is a promising intermediate solution for supporting network-wide interference management when centralized precoding becomes too costly.}
\end{figure}

\subsection{Comparison among different CSIT configurations}
\label{ssec:comparison_information}
We numerically evaluate the Pareto optimal achievable rates $R_k = -\log(\mathrm{MSE}_k(\rvec{t}^\star_k))$, where $\rvec{t}^\star_k$ denotes the optimal solution of Problem \eqref{eq:team_prob}, under the following CSIT configurations: (i) no CSIT sharing \eqref{eq:local_CSIT}, (ii) unidrectional CSIT sharing \eqref{eq:unidirectional_CSIT}, and (iii) full CSIT sharing as in Corollary \ref{cor:bidirectional_TMMSE}. The resulting optimal precoding schemes are respectively denoted by (i) local TMMSE, (ii) unidirectional TMMSE, and (iii) centralized MMSE. We assume for simplicity $\hat{\rvec{H}}_l = \rvec{H}_l$ to study the impact of the different CSIT configurations in absence of measurement noise, and focus on the Pareto optimal point parametrized by $w_k = 1$ $\forall k \in \set{K}$.

Figure \ref{fig:cdf_stadium} reports the empirical cumulative distribution function (CDF) of $R_k$ for multiple i.i.d.  realizations of the RX locations. As expected, adding information constraints on the CSIT configuration leads to performance degradation. However, the degradation is less pronounced from centralized to unidirectional MMSE precoding, showing that unidirectional CSIT sharing does not prevent effective forms of network-wide interference management. Therefore, the unidirectional team MMSE scheme appears as a promising intermediate solution whenever centralized MMSE precoding becomes too costly, e.g., when the CSIT sharing overhead becomes problematic due to high RXs mobility. Quantifying the savings in terms of CSIT sharing overhead is an interesting open problem which depends on many implementation details. For instance, if computational complexity is not an issue and the message sharing is implemented through the forwarding of high-precision  I/Q symbols, unidirectional TMMSE precoding may have the same overhead as local TMMSE precoding, owing to the sequential implementation outlined in Remark~\ref{rem:recursive}. If this is not possible, for instance because non-linear operations such as matrix inversions at each symbol time are not allowed, then the savings may become less prominent, e.g., down to a factor 2.

\begin{figure*}[ht!]
\centering
\subfloat[]{\includegraphics[width=0.4\columnwidth]{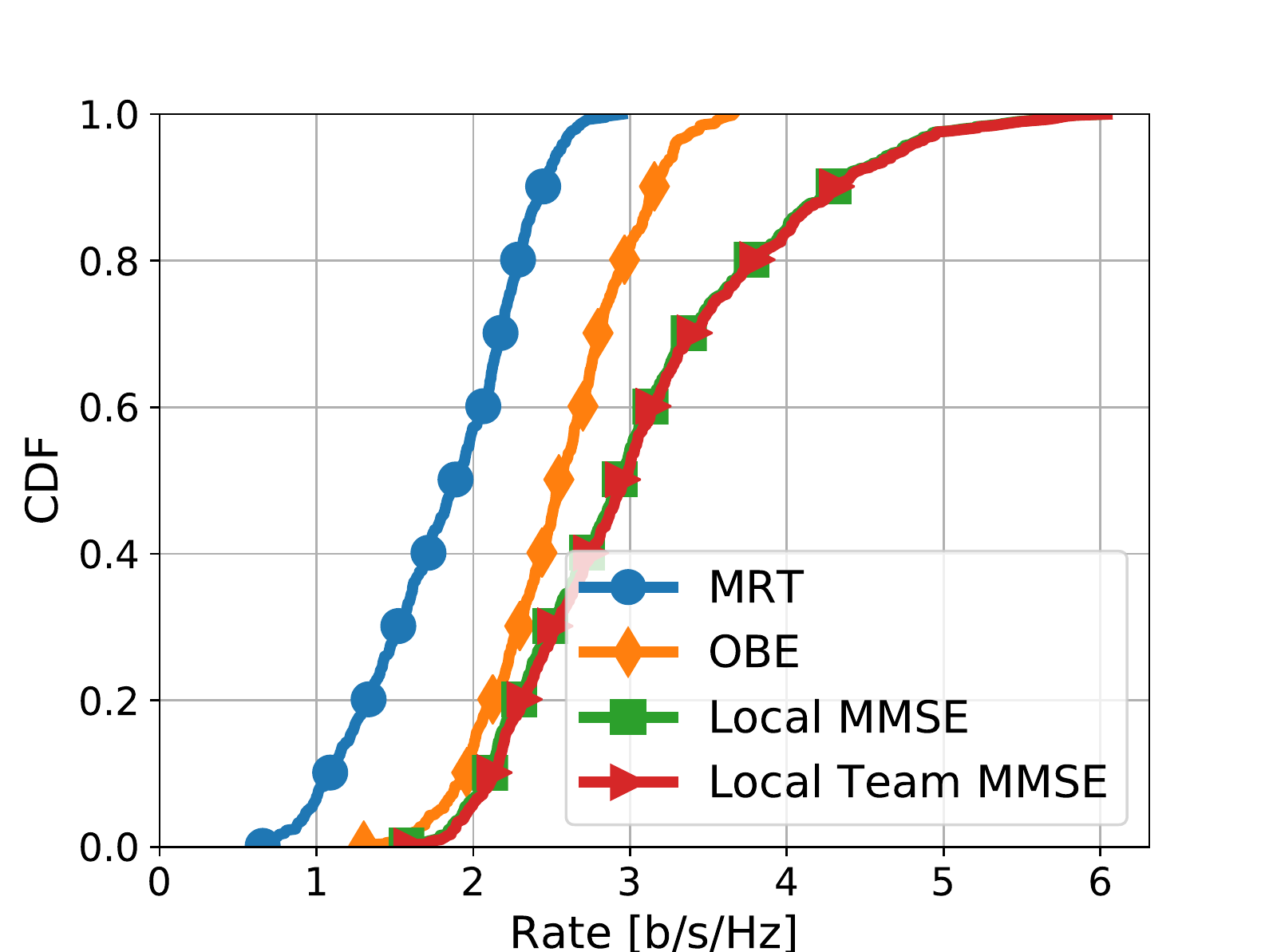}
\label{fig:NLoS}}
\hfil
\subfloat[]{\includegraphics[width=0.4\columnwidth]{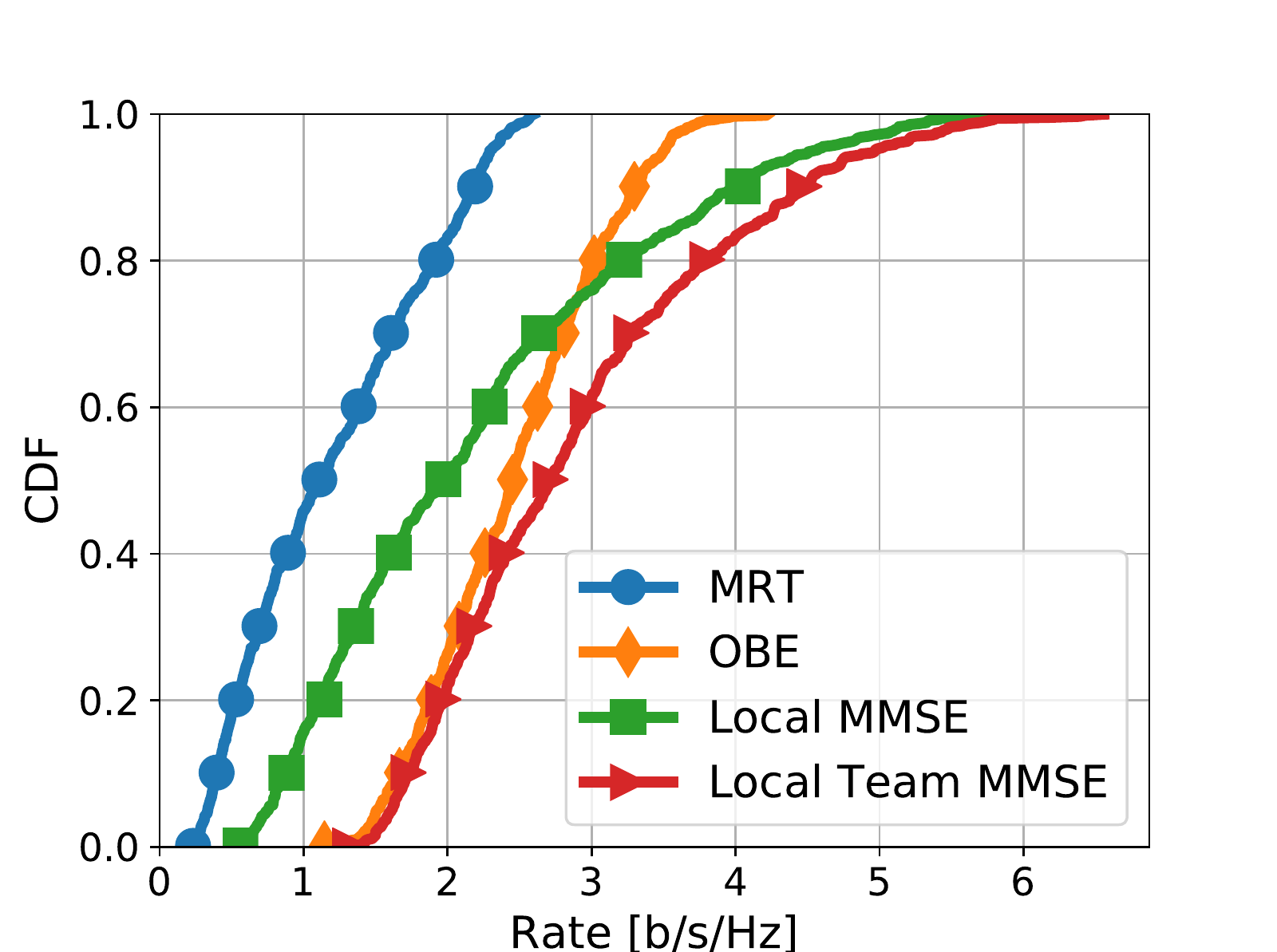}
\label{fig:LoS}}
\caption{Empirical CDF of the per-RX achievable rates for different local precoding schemes, under Ricean factor (a) $\kappa  = 0$, and (b) $\kappa = 1$. In contrast to previously known heuristics, the team MMSE approach optimally exploits statistical information such as the channel mean and always exhibits superior performance. Furthermore, consistently with our theoretical results, local MMSE precoding \cite{emil2020cellfree} is optimal in case (a). However, as expected, and in contrast to the team MMSE approach and the OBE method \cite{neumann2018bilinear}, it may not handle well the interference originating from the channel mean, as shown by the poor perfomance of the weaker RXs in case (b).}
\label{fig:local_precoding}
\end{figure*}
\subsection{Comparison among local precoding schemes}
\label{ssec:comparison_local}
In this section, we compare the optimal local TMMSE solution against classical MRT, the OBE method \cite{neumann2018bilinear}, and local MMSE precoding \eqref{eq:localMMSE} with optimal large-scale fading coefficients $c_{l,k}$ computed using the method in \cite{emil2020cellfree}.  
Since the bound in \eqref{eq:WMSE_lower_bound} may be overly pessimistic for suboptimal schemes, for a fair comparison we compute the DL rates $R_k=R_k^{\mathrm{hard}}$ by means of their dual UL rates $R_k^{\mathrm{UatF}}$ as defined in the proof of Theorem \ref{th:duality_WMSE}, using the same dual UL power allocation $w_k = 1$ $\forall k\in \set{K}$. 
One of the major weaknesses of MRT and local MMSE precoding is that they do not exploit channel mean information, typically arising from LoS components. To study this effect, we modify our simulation setup by letting $N=1$ and by considering a simple Ricean fading model $H_{l,k,1} \sim \CN\left(\sqrt{\frac{\kappa}{\kappa+1}\rho^2_{l,k}}, \frac{1}{\kappa+1}\rho^2_{l,k}\right)$ for some $\kappa\geq 0$, and consider again no measurement noise $\hat{\rvec{H}}_l=\rvec{H}_l$. Figure \ref{fig:local_precoding} confirms the above observation: while, as expected, local MMSE precoding is optimal for a NLoS setup ($\kappa = 0$), it may incur significant performance loss w.r.t. local TMMSE precoding and the OBE method even in case of relatively weak LoS components ($\kappa = 1$).

\subsection{Comparison between unidirectional TMMSE precoding and the SGD scheme \cite{rodriguez2020decentralized}}
In this section, we compare the unidirectional TMMSE solution \eqref{eq:unidirectional_CSIT} for $N=1$ against the suboptimal SGD scheme \eqref{eq:SGD} proposed in \cite{rodriguez2020decentralized} for $\mu_k = 1$ and its robust version obtained by optimizing $\mu_k$ statistically via line search. Figure \ref{fig:TMMSEvsSGD} plots the rate $R_1=R_1^{\mathrm{hard}}$ of the first RX (measured via its dual UL rate as in Section \ref{ssec:comparison_local}) versus the $\mathrm{SNR}:=P\sum_{l}\rho^2_{l,1}$ for a single realization of the simulation setup, and by focusing on the following aspects:
\begin{enumerate}[(a)]
\item Equal path-loss, i.e., $r_2 = 0$ (all RXs colocated at the center of the circular service area), and no channel estimation errors, i.e., $\rvec{E}_l = \vec{0}_{K\times 1}$ $\forall l \in \set{L}$; 
\item Equal path-loss, and channel estimation errors, i.e., we let $\rvec{E}_l \sim \CN(\vec{0}_{K\times 1},\epsilon\vec{K}_l)$ $\forall l \in \set{L}$ and $\hat{\rvec{H}}_l \sim \CN(\vec{0}_{K\times 1},(1-\epsilon)\vec{K}_l)$, where $\vec{K}_l = \mathrm{diag}(\rho_{l,1},\ldots,\rho_{l,K})$ and $\epsilon = 0.2$;
\item Realistic path-loss, i.e., $r_2 = 50$ m (single realization), and no channel estimation errors.
\end{enumerate}
Although the SGD scheme assumes no channel estimation errors, in the above experiments we adapt \eqref{eq:SGD} to  case (b) by replacing $\rvec{h}_l$ with $\hat{\rvec{H}}_l$ everywhere.
\begin{figure*}[!t]
\centering
\subfloat[]{\includegraphics[width=0.31\columnwidth]{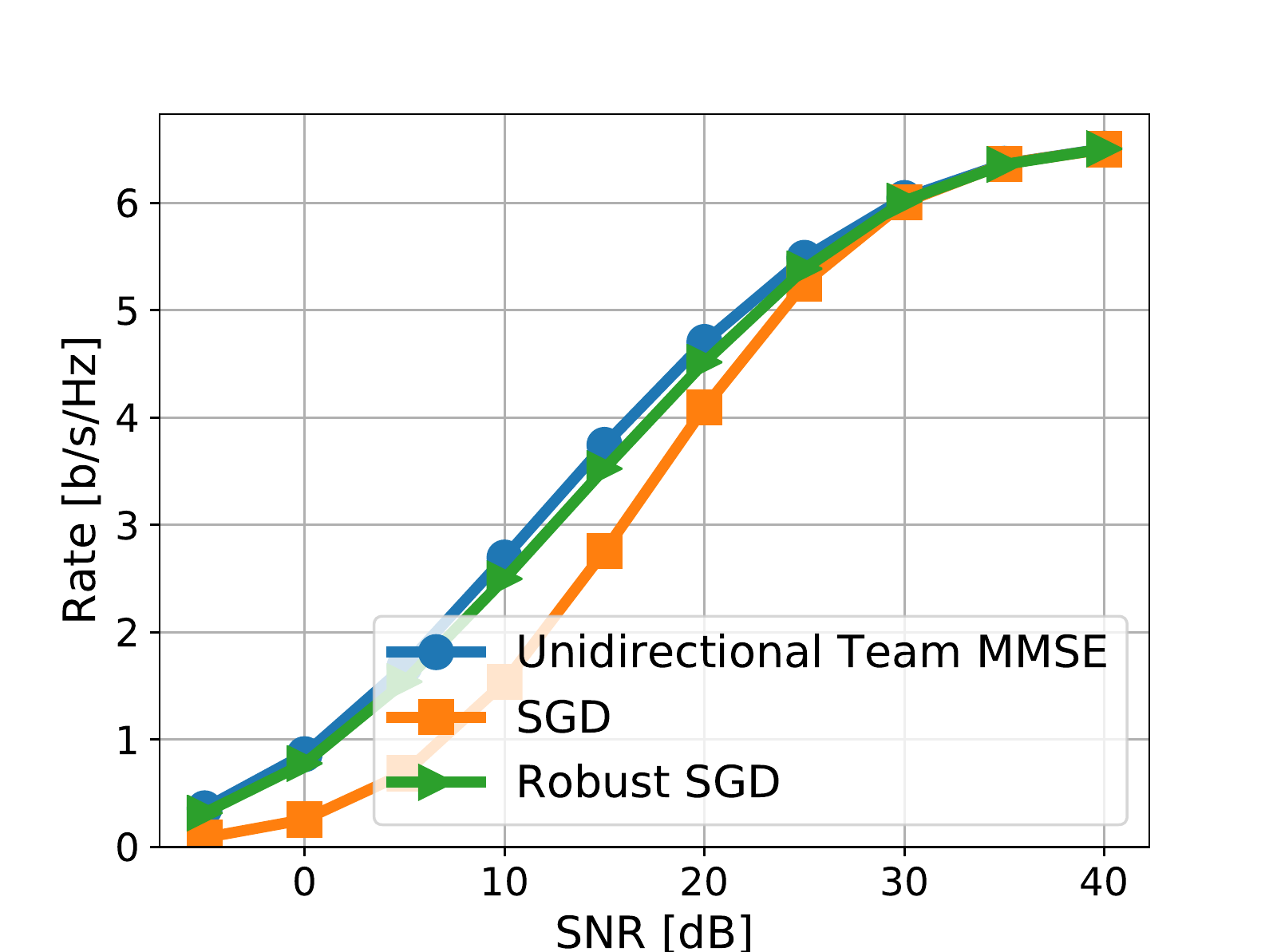}
\label{fig:iid}}
\hfil
\subfloat[]{\includegraphics[width=0.31\columnwidth]{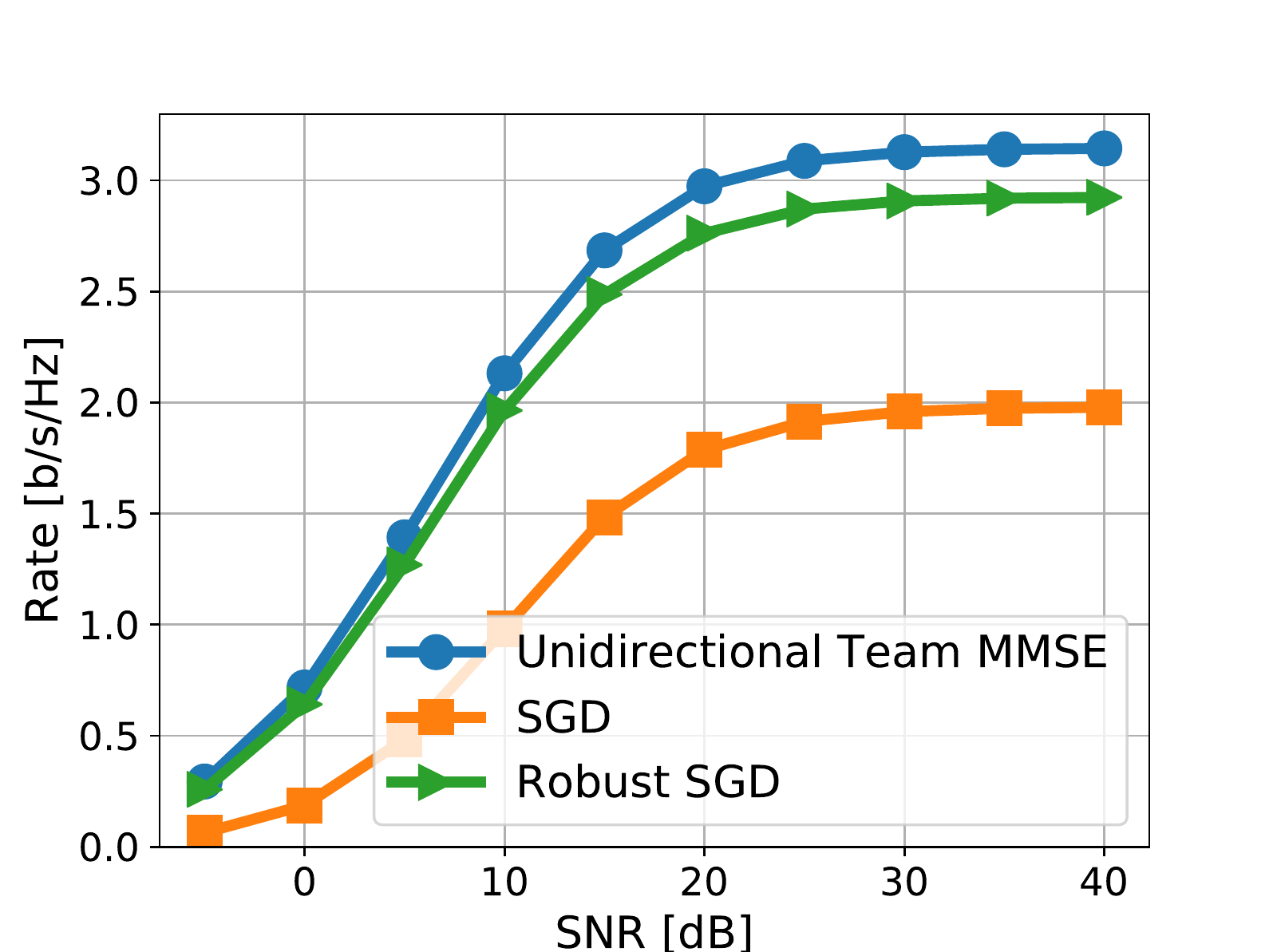}
\label{fig:noise}}
\hfil
\subfloat[]{\includegraphics[width=0.31\columnwidth]{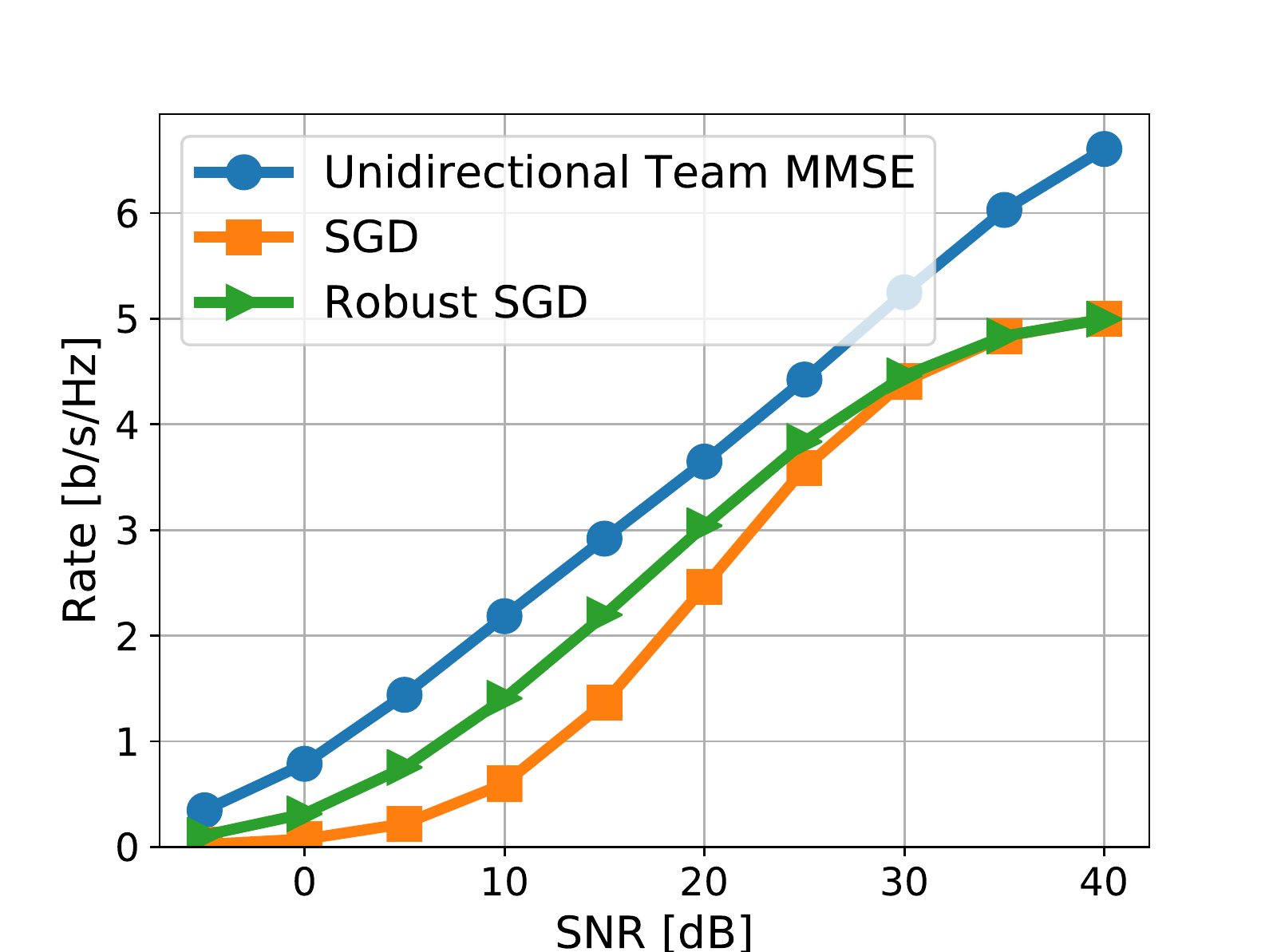}
\label{fig:realistic}}
\caption{Rate vs SNR for RX $1$ under: (a) equal path-loss and no channel estimation errors; (b) equal path-loss and channel estimation errors; (c) realistic path-loss and no channel estimation errors. In contrast to the (robust) SGD scheme \cite{rodriguez2020decentralized}, the team MMSE approach optimally exploits  the path loss information and hence exhibits superior performance in case (c).}
\label{fig:TMMSEvsSGD}
\end{figure*}
As expected, from Figure \ref{fig:iid} we observe that the SGD scheme is asymptotically optimal in case (a), but its performance degrades for low SNR, or in the presence of channel estimation error noise and/or realistic path-loss as shown in Figure \ref{fig:noise} and \ref{fig:realistic}. In contrast, its robust version seems sufficient to recover most of the loss due to finite SNR and channel estimation errors. However, Figure \ref{fig:realistic} shows that the (robust) SGD scheme may not handle more realistic path-loss configurations. 

The main advantage of the (robust) SGD scheme over optimal unidirectional TMMSE precoding is that the former does not perform any $K\times K$ matrix inversion. Therefore, it may be considered as a low-complexity alternative to unidirectional TMMSE precoding. However, further research is needed in particular regarding the choice of the parameters $\mu_{l,k}$ and the support for $N>1$ TX antennas.

\section{Concluding remarks}
This work provides novel guidelines for distributed precoding design in systems with distributed CSIT such as cell-free massive MIMO networks. By assuming full data sharing and a sum-power constraint, the proposed optimal approach outperforms previously known heuristic methods in several setups of interest. Going beyond the two chosen examples where the optimal precoders can be derived explicitly, the proposed approach can be potentially applied, exactly or using standard numerical approximations, to a wide range of practical setups. Furthermore, since it exploits the UL-DL duality principle, we remark that the proposed approach also provides optimal distributed combiners for UL operations. Although not covered for simplicity, the proposed approach may be also extended to limited data sharing, for example using network clustering techniques. Other promising lines of research include the extension to different power constraints, the revisitation of cell-free massive MIMO performance analysis involving LoS models or pilot contamination, and the development of new efficient algorithms exploiting fronthaul architectures such as tree or star topologies.
  
\section*{Acknowledgment}
We would like to thank Prof. Serdar Y\"uksel for his precious technical comments on Theorem \ref{th:quadratic_teams_new}, and Prof. Luca Sanguinetti for his valuable insights on Section \ref{ssec:localCSIT}.

\appendix
\section{Collection of Proofs }
\subsection{Proof of Theorem \ref{th:duality_WMSE}}\label{proof:duality}
Consider a dual UL network with $K$ single-antenna TXs and $L$ cooperating RXs equipped with $N$ antennas each, governed by the MIMO channel law $\rvec{y}^{\mathrm{UL}} = \sum_{k=1}^K\sqrt{Pw_k}\rvec{g}_kV_k+ \rvec{n}^{\mathrm{UL}}$,
where $\rvec{y}^{\mathrm{UL}} \in \stdset{C}^{LN}$ is the received signal at all RXs, $\begin{bmatrix}
\rvec{g}_1,\ldots,\rvec{g}_K
\end{bmatrix}=\rvec{H}^\herm$ is the dual channel matrix, $V_k\sim \CN(0,1)$ is the independent message of TX $k$, and $\rvec{n}^{\mathrm{UL}}\sim \CN(\vec{0},\vec{I})$. Then, we consider the processed channel 
$\hat{V}_k = \frac{1}{\sqrt{P}}\rvec{t}_k^\herm\rvec{y}^{\mathrm{UL}}$, where $\rvec{t}_k=\begin{bmatrix}
\rvec{t}_{1,k}^\T & \ldots & \rvec{t}_{L,k}^\T
\end{bmatrix}^\T $ is a distributed linear combiner satisfying the information constraint $ \rvec{t}_k \in \set{T}$. Let $R_k^{\mathrm{UL}} := I(V_k;\hat{V}_k)$ be achievable rates on this channel. By standard information inequalities \cite{el2011network}, we obtain
\begin{equation}
\begin{split}
I(V_k;\hat{V}_k) &= h(V_k) - h(V_k|\hat{V}_k) \\
&\geq \log(\pi e)-h(V_k-\alpha\hat{V}_k)\\
&\geq \log(\pi e)-\log(\pi e \E[|V_k-\alpha\hat{V}_k|^2]).
\end{split}
\end{equation}
Optimizing $\alpha$ according to channel statistics, i.e., choosing $\alpha = \alpha^\star$ with $\alpha^\star := \E[V_k\hat{V}_k^*]/\E[|\hat{V}_k|^2]$ being the solution of $\min_{\alpha\in \stdset{C}}\E[|V_k-\alpha\hat{V}_k|^2]$, leads to the well-known UatF bound \cite{massivemimobook}
\begin{equation}
R_k^{\mathrm{UL}} \geq R_k^{\mathrm{UatF}}:= \log(1+\mathrm{SINR}_k),
\end{equation}  
\begin{equation}
\mathrm{SINR}_k := \dfrac{w_k|\E[\rvec{t}_k^\herm\rvec{g}_k]|^2}{w_k\V[\rvec{t}_k^\herm\rvec{g}_k]+\sum_{j\neq k} w_j\E[|\rvec{t}^\herm_k\rvec{g}_j|^2]+P^{-1}\E[\|\rvec{t}_k\|^2]}.
\end{equation}
Alternatively, we can keep $\alpha \in \stdset{C}$ unoptimized and obtain the bound $R_k^{\mathrm{UL}} \geq R_k^{\mathrm{UatF}} \geq \log(\E[|V_k-\alpha\hat{V}_k|^2])^{-1}$, where after simple manipulations we recognize 
\begin{equation}
\E[|V_k-\alpha\hat{V}_k|^2] = \E\left[\left\|\alpha\vec{W}^{\frac{1}{2}}\rvec{H} \rvec{t}_k -\vec{e}_k \right\|^2\right] + \frac{\alpha^2}{P}\E[\|\rvec{t}_k\|^2] = \mathrm{MSE}_k(\alpha\rvec{t}_k).
\end{equation}
The above steps also shows that solving $\min_{\alpha \in \stdset{C},\rvec{t}_k \in \set{T}} \mathrm{MSE}_k(\alpha\rvec{t}_k)$ is equivalent to solving
\begin{equation}\label{eq:UatF_maximization}
\underset{\rvec{t}_k \in \set{T}}{\text{maximize}} \; R_k^{\mathrm{UatF}}.
\end{equation}
Furthermore, we observe that $\min_{\alpha \in \stdset{C},\rvec{t}_k \in \set{T}} \mathrm{MSE}_k(\alpha\rvec{t}_k) = \min_{\rvec{t}_k \in \set{T}} \mathrm{MSE}_k(\rvec{t}_k)$. This is because $\alpha$ is a deterministic scalar, hence $\alpha \rvec{t}_k \in \set{T}$. Therefore, Problem \eqref{eq:UatF_maximization} and Problem \eqref{eq:team_prob} have the same optimal solution $\rvec{t}_k^\star$, and the optima are related by $R_k^{\mathrm{UatF}} = \log(\mathrm{MSE}_k(\rvec{t}_k^\star))^{-1}$.

Let $R^\star_k(\vec{w})$ be the optimum of Problem \eqref{eq:UatF_maximization} for some dual UL power allocation policy $\vec{w}:=(w_1,\ldots,w_K) \in \stdset{R}_+^K$. Let then $\set{R}^{\mathrm{UatF}}$ be the union of all rate tuples $(R_1,\ldots,R_K)\in \stdset{R}^K_+$ satisfying $R_k \leq R^\star_k(\vec{w})$ $\forall k \in \set{K}$, where the union is taken over all $\vec{w}$ satisfying $\sum_{k=1}^K Pw_k \leq P_{\mathrm{sum}}$. By definition of $\set{R}^{\mathrm{UatF}}$, its Pareto boundary $\partial\set{R}^{\mathrm{UatF}}$ is composed by rate tuples of the type $(R_1,\ldots,R_K)$, $R_k=R^\star_k(\vec{w})$, achieved by some $\vec{w}$ satisfying $\sum_{k}w_k \leq K$ and by the optimal combiners $\{\rvec{t}_k\}_{k=1}^K$ solving Problem \eqref{eq:team_prob} $\forall k \in \set{K}$. It turns out that it is possible to fully characterize $\partial\set{R}^{\mathrm{UatF}}$ by restricting $\vec{w}\in \set{W}$, i.e., by using all the available power $KP=P_{\mathrm{sum}}$. This is because $R_k^{\mathrm{UatF}}$ is a continuous monotonic increasing functions of $P$, and so is its (finite) supremum over $\rvec{t}_k \in \set{T}$. Furthermore, it can be shown that \textit{all} $\vec{w}\in \set{W}$ induce Pareto optimal rate tuples. This last statement can be proven by contradiction as follows.


Let $\vec{w} \in \set{W}$ and suppose that $(R_1^\star(\vec{w}),\ldots,R_K^\star(\vec{w})) \notin \partial\set{R}^{\mathrm{UatF}}$, i.e., $\exists \vec{w}' \in \set{W}$, $\vec{w}'\neq \vec{w}$, s.t. $R_k^\star(\vec{w}')>R_k^\star(\vec{w})$ $\forall k\in\set{K}$.
We now build an iterative procedure which moves from $\vec{w}$ to $\vec{w}'$ and contradicts the previous supposition. Consider the following sequence of updates $\vec{w}^{(i)}:=(w_1^{(i)},\ldots,w_K^{(i)})$ for $i=0,\ldots,K-1$, where $\vec{w}^{(0)}:=\vec{w}$ and
\begin{equation}
w_k^{(i)} := \begin{cases} w_k' & \text{ if } k \leq i \\
\frac{\sum_{j>i}w'_j}{\sum_{j>i}w_j^{(i-1)}}w_k^{(i-1)} & \text{ if } k > i 
\end{cases}.
\end{equation}
The $i$-th step of the above procedure changes $w_i$ into the target $w_i'$ and scales all weights $w_k$ with $k>i$ by a common factor s.t. the constraint $\vec{w}^{(i)} \in \set{W} $ is not violated. Note that this constraint also implies that $\vec{w}^{(K-1)}=\vec{w}'$ without the need of a $K$-th update. In the following, we use properties of $R_k^\star$ inferred by the fact that $\log_2\left(1+\frac{ax}{bx + cy + d}\right)$ is continuous monotonic increasing in $x \in \stdset{R}_+$ and continuous monotonic decreasing in $y\in \stdset{R}_+$ for any $a \geq 0$, $b\geq 0$, $c\geq 0$, and $d>0$, and so is its supremum over some family of parameters $(a,b,c,d)$. At step $i=1$, assume w.l.o.g. that $w_1^{(1)}=w_1'\geq w_1$, which also implies $w_k^{(1)} = \frac{\sum_{j>i}w'_j}{\sum_{j>i}w_j^{(0)}}w_k\leq w_k$ for $k>1$. In fact, we can always reindex the RXs such that this assumption holds. When going from $x=1$ to $x = \frac{\sum_{j>i}w'_j}{\sum_{j>i}w_j^{(0)}}\leq 1$, and then subsequently from $y=1$ to $y = \frac{w_1'}{w_1}\geq 1$, the function
\begin{equation}
\log\left(1+ \dfrac{w_k|\E[\rvec{t}_k^\herm\rvec{g}_k]|^2x}{(\sum_{j> 1} w_j\E[|\rvec{t}^\herm_k\rvec{g}_j|^2]-w_k|\E[\rvec{t}_k^\herm\rvec{g}_k]|^2)x+w_1\E[|\rvec{t}^\herm_k\rvec{g}_1|^2]y+\E[\|\rvec{t}_k\|^2]/P}\right)
\end{equation}
for $k>1$ is continuous monotonically decreasing, and so is its supremum over $\rvec{t}_k \in \set{T}$. Hence, we have $R_k^\star(\vec{w}^{(1)})\leq R_k^\star(\vec{w}^{(0)})$ $\forall k>1$. At step $i=2$, we assume  w.l.o.g. that $w_2'\geq w_2^{(1)}$ (otherwise we can just properly reindex all RXs $k>2$), and similarly obtain $R_k^\star(\vec{w}^{(2)})\leq R_k^\star(\vec{w}^{(1)})$ $\forall k>2$. By continuing until step $K-1$, we finally obtain $R_K^\star(\vec{w}^{(K-1)})\leq R_K^\star(\vec{w}^{(K-2)})$, which can be combined with the previous steps leading to the desired contradiction $R_K^\star(\vec{w}')\leq R_K^\star(\vec{w})$ up to a possible reindexing, i.e., at least one rate cannot be strictly increased when moving from $\vec{w}$ to $\vec{w}'$.

The proof is concluded by invoking the duality principle between the UatF bound and the hardening bound \cite[Theorem~4.8]{massivemimobook}, which shows that $\set{R}^{\mathrm{UatF}} = \set{R}^{\mathrm{hard}}$ and that for every rate tuple $(R_1^{\mathrm{UatF}},\ldots,R_K^{\mathrm{UatF}})$ achieved by some $\{\rvec{t}_k\}_{k=1}^K$ and $\vec{w}$, there is a rate tuple $(R_1^{\mathrm{hard}},\ldots,R_K^{\mathrm{hard}}) = (R_1^{\mathrm{UatF}},\ldots,R_K^{\mathrm{UatF}})$  achievable by using the same functions $\{\rvec{t}_k\}_{k=1}^K$, and by choosing $p_k = \dfrac{\tilde{p}_kP}{\E[\|\rvec{t}_k\|^2]}$ $\forall k \in \set{K}$ with $\tilde{\vec{p}}:=[\tilde{p}_1,\ldots,\tilde{p}_K]^\T$ being the solution of $
(\vec{D}^{-1}-\vec{B})\tilde{\vec{p}} = (\vec{D}^{-1}-\vec{B}^\T)\vec{w}$, where $\vec{D} := \mathrm{diag}(d_1,\ldots,d_K)$, $d_k:= \mathrm{SINR}_k \dfrac{\E[\| \rvec{t}_k \|^2]}{|\E[\rvec{t}_k^\herm \rvec{g}_{k}]|^2}$, and where the $(k',k)$-th element of $\vec{B} \in \stdset{C}^{K\times K}$ is given by
\begin{equation}
[\vec{B}]_{k',k} = \begin{cases}\dfrac{\E[|\rvec{t}_k^\herm \rvec{g}_{k'}|^2]}{\E[\| \rvec{t}_k \|^2]} & \text{ if } k'\neq k \\
\dfrac{\E[|\rvec{t}_k^\herm \rvec{g}_{k}|^2]-|\E[\rvec{t}_k^\herm \rvec{g}_{k}]|^2}{\E[\| \rvec{t}_k \|^2]} & \text{ otherwise.} \end{cases}
\end{equation}
The above linear system is guaranteed to have a unique solution satisfying $\sum_{k=1}^K \tilde{p}_k = \sum_{k=1}^K w_k$, which implies $\sum_{l=1}^L\E[\|\rvec{x}_l\|^2] = \sum_{k=1}^Kw_kP = P_{\mathrm{sum}}$.

\subsection{Proof of Theorem \ref{th:quadratic_teams_new}}
\label{proof:quadratic_teams_new}
To avoid cumbersome notation, we omit the subscript $k$ everywhere. The proof is split into three separate lemmas. We start with a minor extension of \cite[Theorem~3]{radner1962team} obtained by introducing the constraint $\E[\|\rvec{t}_k\|^2]<\infty$ and specializing to the considered cost function.
\begin{lemma}[Existence and uniqueness]\label{lem:quadratic_teams_original} Problem~\eqref{eq:team_prob} admits a unique team optimal solution.
\end{lemma}
\begin{proof}
Let $\set{H}$ be the space of $\Sigma$-measurable functions $\rvec{a}:\Omega\to \stdset{C}^{LN}$ s.t. $\E\left[\rvec{a}^\herm\rvec{Q}\rvec{a}\right]<\infty$. We define the inner product $\langle \rvec{a},\rvec{b}\rangle := \E[\rvec{b}^\herm\rvec{Q}\rvec{a}]$, $\forall (\rvec{a},\rvec{b}) \in \set{H}^2$, and its induced norm  $\|\rvec{a}\|_{\set{H}} := \sqrt{\langle \rvec{a},\rvec{a}\rangle}$, $\forall \rvec{a} \in \set{H}$. Then, the tuple $(\set{H},\langle\cdot,\cdot \rangle)$ is a  Hilbert space\footnote{In fact, the positive matrix square root $\rvec{Q}^{\frac{1}{2}}$ induces an isometry between $(\set{H},\langle\cdot,\cdot \rangle)$ and the perhaps more familiar Hilbert space of measurable functions such that $\E\left[\|\rvec{a}\|^2\right]<\infty$, equipped with the standard inner product $\langle \rvec{a},\rvec{b}\rangle := \E[\rvec{b}^\herm\rvec{a}]$.} \cite{radner1962team}.  Let us further define $\rvec{t}_0:= \rvec{Q}^{-1}\rvec{g}$, which is the unique minimizer of $c(\vec{H},\vec{t})$ for any realization $\vec{H}$. Firstly, we observe that $\rvec{t}_0 \in \set{H}$, since
\begin{equation}\label{eq:radner_is_valid}
\begin{split}
\|\rvec{t}_0\|^2_{\set{H}}
&= \E \left[\vec{e}^\herm\rvec{H}\left(\rvec{H}^\herm\rvec{H}+P^{-1}\vec{I}\right)^{-1}\rvec{H}^\herm\vec{e}\right]\\
&\leq \E \left[ \mathrm{tr}\left(\rvec{H}\left(\rvec{H}^\herm\rvec{H}+P^{-1}\vec{I}\right)^{-1}\rvec{H}^\herm \right)\right] \\
&= \E \left[\mathrm{tr}\left(\rvec{H}^\herm\rvec{H}\left(\rvec{H}^\herm\rvec{H}+P^{-1}\vec{I}\right)^{-1}\right)\right]\\
&\leq \E \left[\mathrm{tr}\left(\left(\rvec{H}^\herm\rvec{H}+P^{-1}\vec{I}\right)\left(\rvec{H}^\herm\rvec{H}+P^{-1}\vec{I}\right)^{-1}\right)\right] = NL.
\end{split}
\end{equation}
Secondly, we observe that $\set{T}$ is a closed linear subspace of $\set{H}$ \cite[Theorem~2.6.6]{yukselbook}. Finally, simple algebraic manipulations show that the objective of Problem~\eqref{eq:team_prob} can be equivalently rewritten as $\mathrm{MSE}(\rvec{t}) = \|\rvec{t}-\rvec{t}_0\|^2_{\set{H}} -\|\rvec{t}_0\|^2_{\set{H}} + 1$. Therefore, by following \cite{radner1962team, yukselbook}, we consider the infinite dimensional orthogonal projection problem:
\begin{equation}
\label{eq:quadratic_projection}
\underset{\rvec{t}\in \set{T}}{\text{minimize}} \;\|\rvec{t}-\rvec{t}_0\|^2_{\set{H}}.
\end{equation}  
The solution to Problem~\eqref{eq:team_prob} corresponds to the projection of $\rvec{t}_0\in \set{H}$ onto the closed linear subspace $\set{T} \subseteq{\set{H}}$. By the Hilbert projection theorem, this projection is unique and always exists \cite{luenberger1997optimization}.
\end{proof} 

The following result extends Lemma~\ref{lem:quadratic_teams_original} by following similar lines as \cite[Theorem~2.6.6]{yukselbook}.
\begin{lemma}[Sufficiency of stationarity]\label{lem:sufficiency} Suppose that $\E[\|\rvec{Q}\|^2_{\mathrm{F}}]<\infty$.
Then, if $\rvec{t}^\star\in \set{T}$ is stationary, it is also the unique optimal solution to Problem~\eqref{eq:team_prob}.
\end{lemma}
\begin{proof}
Let us consider again the equivalent problem \eqref{eq:quadratic_projection}. Since  $\set{T}$ is a closed linear subspace, the Hilbert projection theorem also states that a solution $\rvec{t}^\star \in \set{T}$ is the unique optimal solution if and only if the following orthogonality conditions \cite{luenberger1997optimization} hold: $(\forall \rvec{t} \in \set{T})$
\begin{equation}\label{eq:orthogonality_condition}
\begin{gathered}
\langle \rvec{t}^\star-\rvec{t}_0, \rvec{t} \rangle = 0, \\
\iff \E\left[\rvec{t}^\herm\rvec{Q}\left(\rvec{t}^\star-\rvec{Q}^{-1}\rvec{g}\right)\right] = 0, \\
\iff\E\left[\sum_{l=1}^L \rvec{t}_l^\herm\left(\E[\rvec{Q}_{l,l}|S_l]\rvec{t}^\star_l + \sum_{j\neq l}\mathbb{E}[\rvec{Q}_{l,j}\rvec{t}^\star_j|S_l] - \E[\rvec{g}_l|S_l]\right) \right] = 0,
\end{gathered}
\end{equation}
where the last equality follows by the law of total expectation, provided that the inner expectations are finite. Finiteness of $\E[\rvec{g}_l|S_l]$ and $\E[\rvec{Q}_{l,l}|S_l]$ follows by the assumption $\E[\|\rvec{H}\|_{\mathrm{F}}^2] < \infty$. Finiteness of $\mathbb{E}[\rvec{Q}_{l,j}\rvec{t}^\star_j|S_l]$ follows by applying the Cauchy-Schwarz inequality elementwise, and by using $\E[\|\rvec{t}_j\|^2] < \infty$ and $\E[\|\rvec{Q}\|_{\mathrm{F}}^2] < \infty$. The proof is concluded by observing that if the stationary conditions in \eqref{eq:stationary_evaluated} are satisfied for some $\rvec{t}^\star$, then the orthogonality conditions are satisfied and $\rvec{t}^\star$ is the unique optimal solution.
\end{proof}

To conclude the proof, it remains to show the converse statement of Lemma \ref{lem:sufficiency}. In the following, we depart from \cite[Theorem~2.6.6]{yukselbook} and use a different argument tailored to the cost function considered in here.
\begin{lemma}[Necessity of stationarity]\label{lem:necessity}
Suppose that $\E[\|\rvec{Q}\|^2_{\mathrm{F}}]<\infty$. Then, if $\rvec{t}^\star\in \set{T}$ is the unique optimal solution to Problem~\eqref{eq:team_prob}, it is also stationary.
\end{lemma}
\begin{proof}
We start by using the so-called notion of \textit{person-by-person optimality} \cite{radner1962team,yukselbook}. Similarly to the notion of Nash equilibrium, it states that a necessary condition for a tuple $\rvec{t}^\star$ to be globally optimal is that it must satisfy
\begin{equation}
\mathrm{MSE}(\rvec{t}^\star) = \underset{\rvec{t}_l \in \set{T}_l}{\text{min}} \mathrm{MSE}(\rvec{t}^\star_{-l},\rvec{t}_l),\quad \forall l \in \set{L}.
\end{equation}
We relax the the above conditions by letting $\set{T}_{l,\mathrm{unc}}$ be the unconstrained version of $\set{T}_l$, i.e., where we remove the constraint $\E[\|\rvec{t}_l\|^2]<\infty$. We then have 
\begin{equation}\label{eq:relaxation_chain}
\begin{split}
\infty > \mathrm{MSE}(\rvec{t}^\star) &= \underset{\rvec{t}_l \in \set{T}_l}{\text{min}} \mathrm{MSE}(\rvec{t}^\star_{-l},\rvec{t}_l) \\
 &\geq  \underset{\rvec{t}_l \in \set{T}_{l,\mathrm{unc}}}{\text{min}} \mathrm{MSE}(\rvec{t}^\star_{-l},\rvec{t}_l)\\
 &= \underset{\rvec{t}_l \in \set{T}_{l,\mathrm{unc}}}{\text{min}} \E\left[\E\left[c\left(\rvec{H},\rvec{t}^\star_{-l},\rvec{t}_l\right)\middle|S_l\right]\right]\\
&\geq \E\left[\underset{\rvec{t}_l \in \set{T}_{l,\mathrm{unc}}}{\text{min}}\E\left[c\left(\rvec{H},\rvec{t}^\star_{-l},\rvec{t}_l\right)\middle|S_l\right]\right]\\
&= \E\left[\E\left[c\left(\rvec{H},\rvec{t}^\star_{-l},\rvec{t}^{\star\star}_l \right)\middle|S_l\right] \right]\\
&= \mathrm{MSE}(\rvec{t}^\star_{-l},\rvec{t}^{\star\star}_l)
\end{split}
\end{equation}
where $\rvec{t}_l^{\star\star}$ is given by the first-order optimality condition $\nabla_{\vec{t}_l}\phi_l(S_l,\vec{t}_l)=\vec{0}$ a.s. applied to the convex function $\phi_l(s_l,\vec{t}_l):=\E\left[c(\rvec{H},\rvec{t}^\star_{-l},\vec{t}_l)\middle|S_l=s_l\right]$. Note that 
\begin{equation}
\rvec{t}^{\star\star}_l(S_l) := (\E[\rvec{Q}_{l,l}|S_l])^{-1}\left(\E[\rvec{g}_l|S_l]-\sum_{j\neq l}\E[\rvec{Q}_{l,j}\rvec{t}^\star_j|S_l]\right) \in \set{T}_{l,\mathrm{unc}},
\end{equation}
because it is given by sums and products of measurable functions (we recall that $\rvec{Q}_{l,l} \succ \vec{0}$), and that all the expectations are finite as discussed in the proof of Lemma~\ref{lem:sufficiency}. Finally, we observe that $\rvec{t}_l^{\star\star} \in \set{T}_l$, i.e., $\E[\|\rvec{t}^{\star\star}_l\|^2]<\infty$, because from the original problem formulation $\eqref{eq:team_prob}$ we notice that $\mathrm{MSE}(\rvec{t}^\star_{-l},\rvec{t}^{\star\star}_l)$ is given by a sum of non-negative terms, one of which is precisely $\frac{1}{P}\E[\|\rvec{t}^{\star\star}_l\|^2]$, and $\mathrm{MSE}(\rvec{t}^\star_{-l},\rvec{t}^{\star\star}_l) \leq \mathrm{MSE}(\rvec{t}^\star)<\infty$. Therefore, the inequalities in \eqref{eq:relaxation_chain} are equalities, and the optimal solution must satisfy $
\mathrm{MSE}(\rvec{t}^\star) = \mathrm{MSE}(\rvec{t}^\star_{-l},\rvec{t}^{\star\star}_l)$, $ \forall l\in\set{L}$. This proves that an optimal solution must satisfy the stationarity conditions given by \eqref{eq:stationary_evaluated}.
\end{proof}

\subsection{Proof of Lemma \ref{lem:bound}}\label{proof:bound}
We use the same notation and definitions as in the proof of Theorem \ref{th:quadratic_teams_new} given in Appendix~\ref{proof:quadratic_teams_new}. The optimality gap can be expressed as follows:
\begin{equation}
\begin{split}
&\mathrm{MSE}(\rvec{t}) - \mathrm{MSE}(\rvec{t}^\star)\\
&= \|\rvec{t}-\rvec{t}_0\|_{\set{H}}^2-\|\rvec{t}^\star-\rvec{t}_0\|_{\set{H}}^2\\
&\overset{(a)}{=} \|\rvec{t}-\rvec{t}^\star\|_{\set{H}}^2\\
&\overset{(b)}{=} \langle \rvec{t}-\rvec{t}^\star,\rvec{t}-\rvec{t}^\star \rangle + \langle \rvec{t}^\star-\rvec{t}_0,\rvec{t}-\rvec{t}^\star \rangle\\
&= \langle \rvec{t}-\rvec{t}_0,\rvec{t}-\rvec{t}^\star \rangle \\
&\overset{(c)}{=}\E\left[\sum_{l=1}^L (\rvec{t}_l-\rvec{t}_l^\star)^\herm\left(\E[\rvec{Q}_{l,l}|S_l]\rvec{t}_l(S_l) + \sum_{j\neq l}\E[\rvec{Q}_{l,j}\rvec{t}_j|S_l] - \E[\rvec{g}_l|S_l]\right) \right],
\end{split}
\end{equation}
where $(a)$ follows from Pythagoras' theorem, $(b)$ from the orthogonality condition $\langle \rvec{t}^\star-\rvec{t}_0,\rvec{a} \rangle = 0$, $\forall \rvec{a} \in \set{T}$ and $\rvec{t}-\rvec{t}^\star \in \set{T}$, and $(c)$ by applying the law of total expectation as in \eqref{eq:orthogonality_condition}. Then, the proof follows from
\begin{equation}
\begin{split}
\|\rvec{t}-\rvec{t}^\star\|_{\set{H}}^2 &= \E[(\rvec{t}-\rvec{t}^\star)^\herm \rvec{z}]\\
&= \langle\rvec{Q}^{-1}\rvec{z},\rvec{t}-\rvec{t}^\star \rangle\\
&\leq \|\rvec{Q}^{-1}\rvec{z}\|_{\set{H}}\|\rvec{t}-\rvec{t}^\star\|_{\set{H}},
\end{split}
\end{equation}
where the last step is the Cauchy–Schwarz inequality, and where we use $\rvec{Q}^{-1} \preceq P\vec{I}$ which ensures $
\|\rvec{Q}^{-1}\rvec{z}\|^2_{\set{H}} = \mathbb{E}[\rvec{z}^\herm\rvec{Q}^{-1}\rvec{z}]\leq P\E[\rvec{z}^\herm\rvec{z}]< \infty$.

\subsection{Proof of Theorem \ref{th:teamMMSE_local} (additional details)}
\label{proof:teamMMSE_local}
We rearrange the system at hand as $(\vec{D}+\vec{U}\vec{\Pi}^\T) \vec{C}= \vec{U}$, where $\vec{C}:= \begin{bmatrix}
\vec{C}_1 & \ldots &\vec{C}_L 
\end{bmatrix}^\T$, $\vec{\Pi}:= \begin{bmatrix}
\vec{\Pi}_1 & \ldots &\vec{\Pi}_L 
\end{bmatrix}^\T$, $\vec{U}:= \begin{bmatrix}
\vec{I}_K & \ldots &\vec{I}_K 
\end{bmatrix}^\T $, $\vec{D}:=\mathrm{diag}(\vec{I}-\vec{\Pi}_1,\ldots,\vec{I}-\vec{\Pi}_L)$. The proof follows if $\vec{D}+\vec{U}\vec{\Pi}^\T$ is invertible, giving the optimal coefficients $\vec{C}=(\vec{D}+\vec{U}\vec{\Pi}^\T)^{-1}\vec{U}$. By Lemma~\ref{lem:woodbury} given in Appendix \ref{ssec:algebra}, $\vec{D}+\vec{U}\vec{\Pi}^\T$ is invertible if both $\vec{D}$ and $\vec{D}^{-1} + \vec{\Pi}^\T\vec{U} = \vec{D}^{-1} + \sum_{l}\vec{\Pi}_l$ are invertible. Standard arguments show that $\vec{0}\preceq \vec{\Pi}_l \prec \vec{I}$. Therefore, $\vec{D}$ is Hermitian positive definite (hence invertible), and so is $\vec{D}^{-1} + \sum_{l}\vec{\Pi}_l$, concluding the proof.

\subsection{Proof of Theorem \ref{th:teamMMSE_undirectional} (additional details)}
\label{proof:teamMMSE_unidirectional}
Assume $\vec{0}\preceq \vec{\Pi}_l \prec \vec{I}$ for a fixed $l\in \set{L}$. Then, let $\rvec{Q}_l := \hat{\rvec{H}}_l^\herm\hat{\rvec{H}}_l + \vec{\Sigma}_l + P^{-1}\vec{I}$ and observe that this implies $
\rvec{Q}_l -\hat{\rvec{H}}_l^\herm \vec{\Pi}_l\hat{\rvec{H}}_l = \hat{\rvec{H}}_l^\herm(\vec{I}- \vec{\Pi}_l)\hat{\rvec{H}}_l + \vec{\Sigma}_l + P^{-1}\vec{I} \succ \vec{0}.$
Therefore, we obtain
\begin{equation}
\begin{split}
\hat{\rvec{H}}_l\left(\rvec{Q}_l - \hat{\rvec{H}}_l^\herm\vec{\Pi}_l\hat{\rvec{H}}_l\right)^{-1}\hat{\rvec{H}}_l^\herm  
&= \hat{\rvec{H}}_l\left[\rvec{Q}_l\left(\vec{I} - \rvec{F}_l\vec{\Pi}_l\hat{\rvec{H}}_l\right)\right]^{-1}\hat{\rvec{H}}_l^\herm \\
&\overset{(a)}{=} \hat{\rvec{H}}_l\left(\vec{I} - \rvec{F}_l\vec{\Pi}_l\hat{\rvec{H}}_l\right)^{-1}\rvec{F}_l \\
&\overset{(b)}{=} \rvec{P}_l \left(\vec{I} - \vec{\Pi}_l\rvec{P}_l\right)^{-1} 
\end{split}
\end{equation}
where $(a)$ and $(b)$ follow from Lemma~\ref{lem:inverse_product} and Lemma~\ref{lem:push_through} given in Appendix \ref{ssec:algebra}, respectively. These lemmas ensure that all the above inverses exist, and in particular $(\vec{I}-\vec{\Pi}_l\rvec{P}_l)^{-1}$. Furthermore, the above chain of equalities also show that
\begin{equation}
\begin{split}
\vec{\Pi}_{l-1} &= \E[\rvec{P}_l\rvec{V}_l] + \vec{\Pi}_l\E[\bar{\rvec{V}}_l] \\
&= \vec{\Pi}_l + (\vec{I}-\vec{\Pi}_l)\E[\rvec{P}_l \left(\vec{I} - \vec{\Pi}_l\rvec{P}_l\right)^{-1} ](\vec{I}- \vec{\Pi}_l)\\
&= \vec{\Pi}_l + (\vec{I}-\vec{\Pi}_l)^{\frac{1}{2}}\E[\tilde{\rvec{P}}_l](\vec{I}-\vec{\Pi}_l)^{\frac{1}{2}},
\end{split}
\end{equation}
where $\tilde{\rvec{P}}_l:= \tilde{\rvec{H}}_l\left(\tilde{\rvec{H}}_l^\herm\tilde{\rvec{H}}_l + \vec{\Sigma}_l + P^{-1}\vec{I}\right)^{-1}\tilde{\rvec{H}}_l^\herm$, and $\tilde{\rvec{H}}_l:= (\vec{I}-\vec{\Pi}_l)^{\frac{1}{2}}\hat{\rvec{H}}_l$. By standard argument, it can be shown that $\vec{0} \preceq \tilde{\rvec{P}}_l \prec \vec{I}$ holds, and hence $\vec{0}\preceq\vec{\Pi}_{l-1}\prec \vec{I}$. Overall, the above discussion proves that $\vec{0}\preceq\vec{\Pi}_l\prec \vec{I}$ implies the existence of $(\vec{I}-\vec{\Pi}_l\rvec{P}_l)^{-1}$ and that $\vec{0}\preceq\vec{\Pi}_{l-1}\prec \vec{I}$. By finally observing that $\vec{0}\preceq\rvec{P}_{L-1} \prec \vec{I}$ and hence $\vec{0}\preceq\vec{\Pi}_{L-1} = \E[\rvec{P}_{L-1}] \prec \vec{I}$, the proof is concluded by repeating the previous argument recursively.

\subsection{Proof of Lemma \ref{lem:bound_high_SNR}}
\label{proof:bound_high_SNR}
We recall the following results from random matrix theory, provided without proof: for $\mathrm{vec}(\rvec{H}_l)\sim \CN(\vec{0},\vec{I})$, we have $\E[\rvec{H}_l(\rvec{H}_l^\herm \rvec{H}_l)^{-1}\rvec{H}_l^\herm] = \frac{N}{K}\vec{I}$ and  $\E[\rvec{H}_l(\rvec{H}_l^\herm \rvec{H}_l)^{-2}\rvec{H}_l^\herm] = \frac{N}{K(K-N)}\vec{I}$. We define the projection matrix $\rvec{P}_l := \rvec{H}_l(\rvec{H}_l^\herm \rvec{H}_l)^{-1}\rvec{H}_l^\herm$ onto $\text{span}(\rvec{H}_l)$, the projection matrix $\rvec{P}_l^\perp := \vec{I}-\rvec{P}_l$ onto its orthogonal complement, and let $\rvec{t}_k = (\rvec{t}_{1,k},\ldots,\rvec{t}_{L,k})$ as in \eqref{eq:sequentialZF}. A simple recursive calculation shows the identity $
\vec{e}_k-\sum_{j=1}^l\rvec{H}_j\rvec{t}_{j,k} = \rvec{P}_l^\perp\rvec{P}_{l-1}^\perp\ldots\rvec{P}_1^\perp\vec{e}_k$.
The first part of the objective in \eqref{eq:team_prob} is then given by  
\begin{equation}
\begin{split}
\E\left[\left\|\vec{e}_k-\sum_{j=1}^L\rvec{H}_j\rvec{t}_{j,k}\right\|^2\right] &=  \vec{e}_k^\herm\E\left[\rvec{P}_1^\perp,\ldots,\rvec{P}_{L-1}^\perp\rvec{P}_L^\perp\rvec{P}_L^\perp\rvec{P}_{L-1}^\perp\ldots\rvec{P}_1^\perp\right]\vec{e}_k\\
&= \left(1-\frac{N}{K}\right)\vec{e}_k^\herm\E\left[\rvec{P}_1^\perp,\ldots,\rvec{P}_{L-1}^\perp\rvec{P}_{L-1}^\perp\ldots\rvec{P}_1^\perp\right]\vec{e}_k\\
&= \left(1-\frac{N}{K}\right)^L,
\end{split}
\end{equation} 
where we used the Hermitian symmetry and idempotency of projection matrices, and the independence between $\rvec{H}_l$ and $\{\rvec{H}_j\}_{j\neq l}$. We now measure the suboptimality of $\rvec{t}_k$ by using Lemma \ref{lem:bound}, specialized to the current setting similarly to Lemma \ref{lem:stationarity_imperfect}. We have:
\begin{equation}
\begin{split}
\rvec{z}_{l,k}(S_l) &= \left(\rvec{H}_l^\herm\rvec{H}_l+P^{-1}\vec{I}\right)\rvec{t}_{l,k}(S_l) + 
\rvec{H}_l^\herm\left( \sum_{j\neq l} \E\left[\rvec{H}_j\rvec{t}_{j,k}\Big|S_l\right] - \vec{e}_k\right) \\
&= P^{-1}\rvec{t}_{l,k}(S_l) - \rvec{H}_l^\herm \E\left[\rvec{P}_L^\perp\rvec{P}_{L-1}^\perp\ldots\rvec{P}_1^\perp\vec{e}_k \Big|S_l \right] \\
&= P^{-1}\rvec{t}_{l,k}(S_l)-\left(1-\frac{N}{K}\right)^{L-l+1}\rvec{H}_l^\herm \rvec{P}_l^\perp\ldots\rvec{P}_1^\perp\vec{e}_k \\
&= P^{-1}\rvec{t}_{l,k}(S_l),
\end{split}
\end{equation}
where the last step follows from the definition of projection matrices, which gives $\rvec{H}_l^\herm \rvec{P}_l^\perp \vec{b} = \vec{0}$  for any $\vec{b}\in \stdset{C}^K$. Furthermore, we have 
\begin{equation}
\begin{split}
\E\left[\left\|\rvec{t}_{l,k}\right\|^2 \right] &= \E\left[\left(\vec{e}_k-\sum_{j=1}^{l-1}\rvec{H}_j\rvec{t}_{j,k}\right)^\herm \rvec{H}_l(\rvec{H}_l^\herm \rvec{H}_l)^{-2}\rvec{H}_l^\herm\left(\vec{e}_k-\sum_{j=1}^{l-1}\rvec{H}_j\rvec{t}_{j,k}\right)\right]\\
&= \dfrac{N}{K(N-K)}\left(1-\frac{N}{K}\right)^{l-1} < \infty.
\end{split}
\end{equation}
Therefore, Lemma \ref{lem:bound} applies and, by using the looser bound in \eqref{eq:bound}, we readily obtain $
\mathrm{MSE}_k(\rvec{t}_k)-\mathrm{MSE}_k(\rvec{t}_k^\star) \leq \frac{1}{P}\sum_{l=1}^L \E\left[\left\|\rvec{t}_{l,k}\right\|^2 \right] \underset{P\to \infty}{\longrightarrow} 0$, and $\mathrm{MSE}_k(\rvec{t}_k^\star) \geq \left(1-\frac{N}{K}\right)^L$.

\subsection{Linear algebra background}
\label{ssec:algebra}
\begin{lemma}[Woodbury matrix identity]\label{lem:woodbury}
Let $\vec{A}\in \stdset{C}^{n\times n}$, $\vec{B} \in \stdset{C}^{n\times m}$, $\vec{C} \in \stdset{C}^{m\times n}$, and $\vec{D} \in \stdset{C}^{m\times m}$. If $\vec{A}$, $\vec{C}$, and $\vec{D}^{-1} + \vec{C}\vec{A}^{-1}\vec{B}$ are invertible, then $\vec{A}+ \vec{B}\vec{D}\vec{C}$ is invertible and $
\left(\vec{A}+ \vec{B}\vec{D}\vec{C}\right)^{-1} = \vec{A}^{-1} - \vec{A}^{-1}\vec{B} \left(\vec{D}^{-1} + \vec{C}\vec{A}^{-1}\vec{B} \right)^{-1}\vec{C} \vec{A}^{-1}$.
\end{lemma}

\begin{lemma}[Inverse of product]\label{lem:inverse_product}
Let $\vec{A}$ and $\vec{B}$ be two square matrices of the same dimension. If $\vec{A}\vec{B}$ is invertible, then $\vec{A}$ and $\vec{B}$ are also invertible, and $
(\vec{A}\vec{B})^{-1}=\vec{B}^{-1}\vec{A}^{-1}$.
\end{lemma}

\begin{lemma}[Push-through identity]\label{lem:push_through}
Let $\vec{A} \in\stdset{C}^{n\times m}$ and $\vec{B}\in \stdset{C}^{m\times n}$ be two matrices such that $\vec{I}+\vec{A}\vec{B}$ is invertible. Then, $\vec{I}+\vec{B}\vec{A}$ is also invertible, and $
\vec{B}\left( \vec{I}+\vec{A}\vec{B}\right)^{-1} =  \left( \vec{I}+\vec{B}\vec{A}\right)^{-1}\vec{B}$.
\end{lemma}

\bibliographystyle{IEEEbib}
\bibliography{refs}

\end{document}